\documentclass[orivec]{llncs}

\usepackage{amssymb}
\usepackage{amsmath}
\usepackage[utf8]{inputenc}
\usepackage{color}
\usepackage{verbatim}
\usepackage{multirow}
\usepackage{tikz}
\usepackage{ltl}
\usepackage{hyperref}
\usetikzlibrary{arrows,positioning,automata}
\usepackage[capitalise,nameinlink]{cleveref}
\usepackage{csquotes}
\usepackage{todonotes}
\usepackage{couriers}
\RequirePackage{listings}
\newcommand{\mb}[1]{\mathbb{#1}}
\newcommand{\mc}[1]{\mathcal{#1}}
\newcommand{\Implies}{\Rightarrow}
\renewcommand{\implies}{\rightarrow}

\newcommand{\A}{\mc{A}}
\newcommand{\B}{\mc{B}}
\newcommand{\T}{\mc{T}}
\newcommand{\R}{\mc{R}}
\newcommand{\G}{\mc{G}}
\newcommand{\M}{\mc{M}}
\newcommand{\I}{\mc{I}}
\renewcommand{\O}{\mc{O}}

\newcommand{\tuple}[1]{(#1)}
\newcommand{\tupleE}[1]{\left\langle #1 \right\rangle}

\title{Bounded Synthesis of Reactive Programs\thanks{Some of the results of the paper are part of the first author's Bachelor thesis~\cite{2017:Gerstacker}.
    Supported by the European Research Council (ERC) Grant OSARES (No. 683300) and by the Saarbrücken Graduate School of Computer Science.}}
\author{Carsten Gerstacker \and Felix Klein \and Bernd Finkbeiner}
\institute{Reactive Systems Group, Saarland University, Germany
  \email{\{gerstacker,fklein,finkbeiner\}@cs.uni-saarland.de}}

\begin{document}

\maketitle

\begin{abstract}
  Most algorithms for the synthesis of reactive systems focus on the construction of
  finite-state machines rather than actual programs. This often leads to badly
  structured, unreadable code. In this paper, we present a bounded
  synthesis approach that automatically constructs, from a given
  specification in linear-time temporal logic (LTL), a program in
  Madhusudan's simple imperative language for reactive programs. We
  develop and compare two principal approaches for the reduction of the
  synthesis problem to a Boolean constraint satisfaction problem. The
  first reduction is based on a generalization of bounded synthesis to
  two-way alternating automata, the second reduction is based on a direct
  encoding of the program syntax in the constraint system.
  We report on preliminary experience with a prototype implementation,
  which indicates that the direct encoding outperforms the
  automata approach.
\end{abstract}

\section{Introduction}
In reactive synthesis, we automatically construct a reactive system,
such as the controller of a cyberphysical system, that is guaranteed
to satisfy a given specification. The study of the synthesis problem,
known also as Church's problem~\cite{Church/57/Applications}, dates back to the 1950s
and has, especially in recent years, attracted a lot
of attention from both theory and practice. There is a growing number
of both tools (cf. \cite{Jobstm07c,DBLP:conf/tacas/Ehlers11,DBLP:conf/cav/BohyBFJR12,DBLP:conf/cav/FaymonvilleFT17}) and success stories, such as the synthesis
of an arbiter for the AMBA AHB bus, an open industrial standard for
the on-chip communication and management of functional blocks in
system-on-a-chip (SoC) designs~\cite{Bloem+others/07/Automatic}.

The practical use of the synthesis tools has, however, so far been limited.
A serious criticism is that, compared to code produced by a human programmer, the code produced by
the currently available synthesis tools is usually badly structured and, quite simply,
unreadable. The reason is that the synthesis tools do not actually
synthesize \emph{programs}, but rather much simpler computational
models, such as \emph{finite state machines}.  As a result, the
synthesized code lacks control structures, such as \emph{while}
loops, and symbolic operations on program \emph{variables}:
everything is flattened out into a huge state graph.

A significant step towards better implementations has been the
\emph{bounded synthesis}~\cite{2013:Schewe:BS} approach, where the number of states
of the synthesized implementation is bounded by a constant.
This can be used to construct finite state machines with a
\emph{minimal} number of states. Bounded synthesis has also been extended with
other structural measures, such as the number of
cycles~\cite{2016:Klein:BCS}.  Bounded synthesis reduces the synthesis
problem to a constraint satisfaction problem: the existence of an
implementation of bounded size is expressed as a set of Boolean
constraints, which can subsequently be solved by a SAT or QBF solver~\cite{DBLP:conf/tacas/FaymonvilleFRT17}.
Bounded synthesis has proven highly effective in finding finite state
machines with a \emph{simple} structure.  However, existing methods based on bounded synthesis do not
make use of syntactical program constructs like loops or variables. 
The situation is different in the synthesis of sequential programs,
where programs have long been studied as the target of synthesis
algorithms~\cite{conf/fmcad/AlurBJMRSSSTU13,conf/popl/Gulwani11,osera2015type,solarLezama13,vechevYY13}.
In particular, in \emph{syntax-guided
synthesis}~\cite{conf/fmcad/AlurBJMRSSSTU13}, the output of the
synthesis algorithm is constrained to programs whose syntax conforms
to a given grammar.  A first theoretical step in this direction for
reactive systems was proposed by
Madhusudan~\cite{2011:Madhusudan:SRP}. Madhusudan defines a small
imperative programming language and shows that the existence of a
program in this language with a fixed set of Boolean variables is
decidable.  For this purpose, the specification is translated into an alternating
two-way tree automaton that reads in the syntax tree of a program,
simulates its behavior, and accepts all programs whose behavior
satisfies the specification. Because the set of variables is fixed in
advance, the approach can be used to synthesize programs with a
minimal number of variables. However, unlike bounded synthesis, this
does not lead to programs that are minimial in other ways, such as the
number of states or cycles.

In this paper, we present the first bounded synthesis approach
for reactive programs. As in standard bounded synthesis~\cite{2013:Schewe:BS}, we
reduce the synthesis problem to a constraint satisfaction problem.
The challenge is to find a constraint system that encodes the
existence of a program that satisfies the specification, and that, at the same time, can be solved
efficiently. We develop and compare two principal methods.
The first method is inspired by Madhusudan's construction in that we
also build a two-way tree automaton that recognizes the correct
programs. The key difficulty here is that the standard bounded synthesis
approach does not work with two-way automata, let alone the
alternating two-way automata produced in Madhusudan's construction.
We first give a new automata construction that produces universal, instead of
alternating, two-way automata. We then generalize bounded synthesis to work
on arbitrary graphs, including the run graphs of two-way automata.
The second method follows the original bounded synthesis approach
more closely. Rather than simulating the execution of the program
in the automaton, we encode the existence of both the program and
its run graph in the constraint system. The correctness of the
synthesized program is ensured, as in the original approach, with
a universal (one-way) automaton derived from the specification.
Both methods allow us to compute programs that satisfy the given specification and that are minimal in measures such as the size of the program.
The two approaches compute the exact same reactive programs, but differ,
conceptually, in how much work is done via an automata-theoretic
construction vs. in the constraint solving. In the first approach,
the verification of the synthesized program is done by the automaton,
in the second approach by the constraint solving. Which approach is better?
While no method has a clear theoretical advantage over the other, our experiments with a
prototype implementation indicate a strong advantage for the second approach.

\section{Preliminaries}
We denote the Boolean values $\{0,1\}$ by $\mb{B}$. The set of non-negative integers is denoted by $\mb{N}$ and for $a\in \mb{N}$ the set $\{0,1,\hdots,a\}$ is denoted by $[a]$. An \textit{alphabet} $\Sigma$ is a non-empty finite set of symbols. The elements of an alphabet are called letters. A \textit{infinite word} $\alpha$ over an alphabet $\Sigma$ is a infinite concatenation $\alpha=\alpha_0\alpha_1\hdots$ of letters of $\Sigma$. The set of infinite words is denoted by $\Sigma^\omega$. With $\alpha_n\in\Sigma$ we access the $n$-th letter of the word. For an infinite word $\alpha\in\Sigma^\omega$ we define with Inf$(\alpha)$ the set of states that appear infinitely often in $\alpha$. A subset of $\Sigma^\omega$ is a \textit{language over infinite words}.

\subsection{Implementations}
\textit{Implementations} are arbitrary input-deterministic reactive systems. We fix the finite input and output alphabet $\I$ and $\O$, respectively. A \textit{Mealy machine} is a tuple $\M = (\I,\O,M,m_0,\tau, o)$ where $\I$ is an input-alphabet, $\O$ is an output-alphabet, $M$ is a finite set of states, $m_0 \in M$ is an initial state, $\tau: M \times 2^\I \rightarrow M $ is a transition function and $o: M \times 2^\I \rightarrow 2^\O$ is an output function. A \textit{system path} over an infinite input sequence $\alpha^\I$ is the sequence $m_0m_1\hdots\in M^\omega$ such that $\forall i\in\mb{N}:\tau(m_i,\alpha^\I_i)=m_{i+1}$. The thereby produced infinite output sequence is defined as $\alpha^\O=\alpha^\O_0\alpha^\O_1\hdots\in(2^\O)^\omega$, where every element has to match the output function, i.e., $\forall i\in\mb{N}: \alpha^\O_i=o(m_i,\alpha^\I_i)$. We say a Mealy machine $\M$ produces a word $\alpha=(\alpha_0^\I\cup\alpha_0^\O)(\alpha_1^\I\cup\alpha_1^\O)\hdots\in (2^{\I\cup\O})^\omega$, iff the output $\alpha^\O$ is produced for input $\alpha^\I$. We refer to the set of all producible words as the language of $\M$, denoted by $\mc{L}(\M)\subseteq(2^{\I\cup\O})^\omega$.

A more succinct representation of implementations are \textit{programs}. The programs we are working with are imperative reactive programs over a fixed set of Boolean variables $B$ and fixed input/output aritys $N_\I$/$N_\O$. Our approach builds upon \cite{2011:Madhusudan:SRP} and we use the same syntax and semantics. Let $b\in B$ be a variable and both $\vec{b_\mc{I}}$ and $\vec{b_\mc{O}}$ be vectors over multiple variables of size $N_\mc{I}$ and $N_\mc{O}$, respectively. The syntax is defined with the following grammar
\begin{center}
\begin{tabular}{lcl}
	$\tupleE{stmt}$ & ::= & $\tupleE{stmt};\tupleE{stmt} \mid$ \textbf{skip} $\mid$ b := $\tupleE{expr}$ $\mid$ \textbf{input} $\vec{b_\mc{I}}$ $\mid$ \textbf{output} $\vec{b_\mc{O}}$ \\
	& $\mid$ & \textbf{if}($\tupleE{expr}$)\textbf{ then} \{$\tupleE{stmt}$\}\textbf{ else }\{$\tupleE{stmt}$\} $\mid$ \textbf{while}($\tupleE{expr}$)\{$\tupleE{stmt}$\}\\
	$\tupleE{expr}$ & ::= & b $\mid$ \textbf{tt} $\mid$ \textbf{ff} $\mid$ ($\tupleE{expr}\vee\tupleE{expr})$ $\mid$ ($\neg\tupleE{expr}$)
\end{tabular}
\end{center}
The semantics are the natural one. Our programs start with an initial variable valuation we define to be $0$ for all variables. The program then interacts with the environment by the means of input and output statements, i.e., for a vector over Boolean variables $\vec{b}$ the statement \enquote{\textbf{input} $\vec{b}$} takes an input in $\{ 0,1 \}^{N_\I}$ from the environment and updates the values of $\vec{b}$. The statement \enquote{\textbf{output} $\vec{b}$} outputs the values stored in $\vec{b}$, that is an output in $\{0,1\}^{N_\O}$. Therefor a program with input/output arity $N_\I/N_\O$ requires at least $max(N_\I,N_\O)$ many variables, i.e., $|B|\geq max(N_\I,N_\O)$. Between two input and output statements the program can internally do any number of steps and manipulate the variables using assignments, conditionals and loops. Note that programs are input-deterministic, i.e., a program maps an infinite input sequence $\alpha^\I\in (\{0,1\}^{N_\I})^\omega$ to an infinite output sequence $\alpha^\O\in (\{0,1\}^{N_\O})^\omega$ and we say a program can produce a word $\alpha=(\alpha^\I_0\alpha^\O_0)(\alpha^\I_1\alpha^\O_1)\hdots\in (\{0,1\}^{N_\I+N_\O})^\omega$, iff it maps $\alpha^\I$ to $\alpha^\O$. We define the language of $\T$, denoted by $\mc{L}(\T)$, as the set of all producible words. We assume programs to alternate between input and output statements.

\begin{figure}[t]
  \centering
  \vspace{-1em}
\quad \begin{minipage}{0.3\textwidth}
		\begin{lstlisting}[mathescape=true]
while(tt) {
  input (r$_1$, r$_2$);
  if(r$_1$) then {
    r$_2$ = ff
  } else {
    skip
  };
  output (r$_1$, r$_2$)
}
		\end{lstlisting}
		\caption{Example-Code}\label{arbiter-tree-code}
\end{minipage}
\begin{minipage}{0.50\textwidth}
	\centering
	\resizebox{0.50\textwidth}{!}{
		\begin{tikzpicture}[->,>=stealth',shorten >=1pt,auto,node distance=2.8cm,thick]
					\tikzstyle{every state}=[rectangle,rounded corners,fill=none,draw=black,text=black,thick,initial text=,inner sep=5pt]
					
					\node[initial,state] 	(A)                		{\Large \textbf{while}};
					\node[state]         	(B) [below left of=A] 	{\Large \textbf{tt}};
					\node[state]			(C) [below right of=A] 	{\Large \textbf{;}};
					\node[state]			(D) [below left of=C]   {\Large \textbf{input} $r_1r_2$};
					\node[state]			(K) [below right of=C]  {\Large \textbf{;}};
					\node[state]			(L) [below right of=K]  {\Large \textbf{output } $r_1r_2$};
					\node[state]			(E) [below left of=K]	{\Large \textbf{if}};
					\node[state]			(F) [below left of=E]   {\Large $r_1$};
					\node[state]			(G) [below right of=E]	{\Large \textbf{then}};
					\node[state]			(H) [below left of=G]   {\Large assign$_{r_2}$};
					\node[state]			(I) [below left of=H]	{\Large \textbf{ff}};
					\node[state]			(J) [below right of=G]  {\Large \textbf{skip}};
					
					\path 
					(A) edge (B)
					(A) edge (C)
					(C) edge (D)
					(C) edge (K)
					(E) edge (F)
					(E) edge (G)
					(G) edge (H)
					(H) edge (I)
					(G) edge (J)
					(K) edge (E)
					(K) edge (L);
					\end{tikzpicture}
	}	
	\caption{Example-Program-Tree}\label{arbiter-tree}
      \end{minipage}
\end{figure}

We represent our programs as $\Sigma$-\textit{labeled binary trees}, i.e., a tuple $\tuple{T, \tau}$ where $T\subseteq\{L,R\}^*$ is a finite and prefix closed set of nodes and $\tau: T \rightarrow \Sigma$ is a labeling function. Based on the defined syntax, we fix the set of labels as 
\begin{center}
	$\Sigma_P = \{ \neg, \vee,;,\text{\textbf{if}},\text{\textbf{then}},\text{\textbf{while}} \} \cup B \cup \{ assign_b \mid b \in B \}$\\ $\cup \{ \text{\textbf{input} } \vec{b} \mid \vec{b} \in B^{N_\I} \} \cup \{ \text{\textbf{output} } \vec{b} \mid \vec{b} \in B^{N_\O} \}$.
\end{center}
We refer to $\Sigma_P$-labeled binary trees as \textit{program trees}. If a node has only one subtree we define it to be a the left subtree. Note that our program trees do therefore not contain nodes with only a right subtree. For example, \cref{arbiter-tree-code} depicts an arbitrary program and \cref{arbiter-tree} the corresponding program tree.

We express the current variable valuation as a function $s: B \rightarrow \mathbb{B}$. We update variables $\vec{b}\in B^n$ with new values $\vec{v}\in \mathbb{B}^n$ using the following notation:
\begin{equation*}
s[\vec{b}/\vec{v}](x) = 
\begin{cases}
v_i &\quad \text{if } b_i = x, \text{for all } i\\
s(x) &\quad \text{otherwise}
\end{cases}
\end{equation*}

\subsection{Automata}
We define \textit{alternating automata over infinite words} as usual, that is a tuple $A=(\Sigma, Q, q_0, \delta, Acc)$ where $\Sigma$ is a finite alphabet, $Q$ is a finite set of states, $q_0 \in Q$ is an initial state, $\delta : Q \times \Sigma \rightarrow \mb{B}^+(Q)$ is a transition function and $Acc \subseteq Q^\omega$ is an acceptance condition.

The \textit{Büchi acceptance condition} $\text{BÜCHI}(F)$ on a set of states $F\subseteq Q$ is defined as $\text{BÜCHI}(F)=\big\{q_0q_1\hdots\in Q^\omega\mid \text{Inf}(\alpha) \cap F \not = \emptyset\big\}$ and $F$ is called the set of accepting states.
The \textit{co-Büchi acceptance condition} $\text{COBÜCHI}(F)$ on a set of states $F\subseteq Q$ is defined as
$\text{COBÜCHI}(F)=\big\{q_0q_1\hdots\in Q^\omega\mid \text{Inf}(\alpha) \cap F = \emptyset\big\}$, where $F$ is called the set of rejecting states.
To express combinations of Büchi and co-Büchi expressions we use the \textit{Streett acceptance condition}. Formally, $\text{STREETT}\big(F\big)$ on a set of tuples $F=\{(A_i,G_i)\}_{i\in[k]}\subseteq Q\times Q$ is defined as
$\text{STREETT}(F)=\big\{q_0q_1\hdots\in Q^\omega\mid \forall i\in[k]:\text{Inf}(\alpha) \cap A_i \not = \emptyset \implies \text{Inf}(\alpha) \cap G_i \not = \emptyset\big\}$. A run with a Streett condition is intuitively accepted, iff for all tuples $(A_i,G_i)$, the set $A_i$ is hit only finitely often or the set $G_i$ is hit infinitely often.

Two-way alternating tree automata are tuple $(\Sigma, P, p_0, \delta _L, \delta _R, \delta _{LR}, \delta _{\emptyset}, Acc)$, where $\Sigma$ is an input alphabet, $P$ is a finite set of states, $p_0 \in P$ is an initial state, $Acc$ is an acceptance condition, and $\delta$ are transition functions of type \\
$\delta_S : P \times \Sigma \times (S \cup \{D\}) \rightarrow \mathbb{B}^+(P \times (S \cup \{U\}))$, for $S \in \{L, R, LR, \emptyset\}$.
We introduce $\mu: T \times \{L,R,U\} \rightarrow T \times \{L,R,D\}$ as a function to map states and directions to move in, to the reached states and the matching incoming directions.
\begin{center}
	\begin{tabular}{llcll}
		$\mu(t,L)$ &$= (t\cdot L, D)$&\mbox{\qquad\qquad}&$\mu(t.L,U)$ &$= (t, L)$\\
		$\mu(t,R)$ &$= (t\cdot R, D)$&&$\mu(t.R,U)$ &$= (t, R)$\\
	\end{tabular}
\end{center}

We consider specifications given in linear time-temporal logic (\textit{LTL}). Such specifications can be translated into non-deterministic Büchi automata or dually into an universal co-Büchi automata as shown in \cite{1994:Vardi:IC}.
For an arbitrary specification we denote by $A_{spec}$ and $\overline{A}_{spec}$ the corresponding non-deterministic Büchi and universal co-Büchi automaton, respectively.

\section{Automata Construction}
We have already argued that programs, as a more succinct representation of implementations, are highly desirable. However, in contrast to Mealy machines, which only dependent on the current state and map an input to a corresponding output, in programs such a direct mapping is not possible. Instead, programs need to be simulated, variables to be altered, expressions to be evaluated and an output statement to be traversed until we produce the corresponding output to the received input. These steps not only depend on the current position in the program but additionally also on the valuation of all variables. 

We build upon Madhusudans reactive program synthesis approach \cite{2011:Madhusudan:SRP} were program synthesis is solved by means of two-way alternating Büchi tree automata walking up and down over program trees while keeping track of the current valuation and the state of a given Büchi specification automaton, which is simulated by the input/output produced by traversing the program tree. The automaton accepts a program tree whenever the simulated specification automaton accepts the provided input/output. The constructed automaton, we will further refer to as $\A$, is intersected with two other constructed automata which enforce syntactically correctness and reactivity of the synthesized program, respectively. Then a reactive and syntactically correct program is synthesized by means of an emptiness check of the obtained automaton, involving an exponential blowup to eliminate two-wayness and alternation.

\subsection{Two-way Universal Co-Büchi Tree Automaton}
We construct a two-way \textit{non-deterministic} Büchi tree automaton $\B$ that is equivalent to $\A$ by using deterministic evaluation of Boolean expressions. We construct $\B$ without an exponential blowup in the state space. We then complement $\B$ into a two-way universal co-Büchi tree automaton convenient for the bounded synthesis approach.

The two-way alternating Büchi tree automaton $\A$ uses universal choices only in relation to Boolean expression evaluation. For example, for $\textbf{if}$, $\textbf{while}$ and $assign_b$-statements a Boolean evaluation is needed. In this cases it non-deterministically guesses whether the expression evaluates to $0$ or $1$ and then universally sends one copy into the Boolean expression, which evaluates to $true$ iff the expression evaluates to the expected value, and one copy to continue the corresponding normal execution. The copy evaluating the Boolean expression walks only downwards and since the subtree corresponding to the Boolean expression is finite, this copy terminates to either $true$ or $false$ after finitely many steps.
Instead of using both non-deterministic and universal choices, we evaluate the Boolean subtree deterministically in finitely many steps and then continue the normal execution based on the result of the evaluation.

Note that we not only remove all universal choices but additionally all unnecessary sources of non-determinism. Therefore, besides traversing input- and output-labels, that introduce unavoidable non-determinism, our program simulation is deterministic.

Our automaton $\B$ with the set of states
$$ P_{exec} = S \times Q_{spec} \times \mb{B}^{N_\I} \times \{inp, out\} \times \mb{B} $$
$$ P^\B_{expr} = S \times Q_{spec} \times \mb{B}^{N_\I} \times \{inp, out\} \times \{ \top , \bot \}$$
$$ P^\B =  P^\B_{expr}\cup P_{exec} $$
and initial state $p^\B_0 = (s_0, q_0, i_0, inp, 0)$, is defined with the transitions shown in \cref{semantics}, where $s\in S$ is a variable valuation, $q\in Q_{spec}$ the state of the simulated specification automaton, $i\in \mb{B}^{N_\I}$ the last received input, $m\in\{inp,out\}$ a flag to ensure alternation between inputs and outputs, $r\in\{\top,\bot\}$ the result of a Boolean evaluation and $t\in\{0,1\}$ a flag for the Büchi condition, which ensures that the specification automaton is simulated for infinite steps and is only set to $1$ for a single simulation step after an output statement. We express states corresponding to Boolean evaluations and program execution as \mbox{$(s,q,i,m,r)\in P^\B_{expr}$} and $(s,q,i,m,t)\in P^\B_{exec}$, respectively.

The notation reads as follows: If the automaton enters a node with one of the black incoming edges, it can move in the direction of the black outgoing edges, while updating his state corresponding to the annotated update expression, depicted by an enclosing rectangle. Additionally, the automaton needs to fulfill the conditions annotated to the edges it traverses. To express non-determinism we use sets of update expressions, such that each expression represents one possible successor.
All state values not contained in the update expression stay the same, except $t$ which is set to $0$. When changing from Boolean evaluation to program execution, we copy $s$, $q$, $i$, $m$ and vice versa.

\begin{figure}
	\centering
	\scalebox{.9}{
	\begin{tikzpicture}[auto,semithick
					virtual/.style={thick,densely dashed},
					active/.style={thick,->,shorten >=2pt,shorten <=2pt,>=stealth},
					inactive/.style={color=lightgray,->,shorten >=2pt,shorten <=2pt,>=stealth},
					none/.style={draw=none}
					]
					\tikzstyle{every state}=[rectangle,rounded corners,fill=none,draw=black,text=black,initial text=]
					\tikzstyle{block} = [rectangle,draw=black,text=black,initial text=, minimum size=.6cm]
					
					\node[state] 	(A) at (0,0) {\textbf{tt}};
					\node[block]    [align=left, anchor=west,xshift=.3cm](B) at (A.east) {$r\rightarrow \top$};
					
					\draw[active] ([xshift=.08cm]A.north) -- (0.08,1.2);
					\draw[active] (-0.08,1.2) -- ([xshift=-.08cm]A.north);
					
					\draw[none] (A.217) -- (-1.05,-.95);
					\draw[none] (-.95,-1.05) -- (A.233);
					
					\draw[none] (A.307) -- (.95,-1.05);
					\draw[none] (1.05,-.95) -- (A.323);
					\end{tikzpicture}
	\begin{tikzpicture}[auto,semithick
					virtual/.style={thick,densely dashed},
					active/.style={thick,->,shorten >=2pt,shorten <=2pt,>=stealth},
					inactive/.style={color=lightgray,->,shorten >=2pt,shorten <=2pt,>=stealth},
					none/.style={draw=none}
					]
					\tikzstyle{every state}=[rectangle,rounded corners,fill=none,draw=black,text=black,initial text=]
					\tikzstyle{block} = [rectangle,draw=black,text=black,initial text=, minimum size=.6cm]
					
					\node[state] 	(A) at (0,0) {\textbf{ff}};
					\node[block]    [align=left, anchor=west,xshift=.3cm](B) at (A.east) {$r\rightarrow \bot$};
					
					\draw[active] ([xshift=.08cm]A.north) -- (0.08,1.2);
					\draw[active] (-0.08,1.2) -- ([xshift=-.08cm]A.north);
					
					\draw[none] (A.217) -- (-1.05,-.95);
					\draw[none] (-.95,-1.05) -- (A.233);
					
					\draw[none] (A.307) -- (.95,-1.05);
					\draw[none] (1.05,-.95) -- (A.323);
					\end{tikzpicture}
	\begin{tikzpicture}[auto,semithick
					virtual/.style={thick,densely dashed},
					active/.style={thick,->,shorten >=2pt,shorten <=2pt,>=stealth},
					inactive/.style={color=lightgray,->,shorten >=2pt,shorten <=2pt,>=stealth},
					none/.style={draw=none}
					]
					\tikzstyle{every state}=[rectangle,rounded corners,fill=none,draw=black,text=black,initial text=]
					\tikzstyle{block} = [rectangle,draw=black,text=black,initial text=, minimum size=.6cm]
					
					\node[state] 	(A) at (0,0) {b};
					\node[block]    [align=left, anchor=west,xshift=.3cm](B) at (A.east) {$r\rightarrow s[b]$};
					
					\draw[active] ([xshift=.08cm]A.north) -- (0.08,1.2);
					\draw[active] (-0.08,1.2) -- ([xshift=-.08cm]A.north);
					
					\draw[none] (A.217) -- (-1.05,-.95);
					\draw[none] (-.95,-1.05) -- (A.233);
					
					\draw[none] (A.307) -- (.95,-1.05);
					\draw[none] (1.05,-.95) -- (A.323);
					\end{tikzpicture}}\\
	\scalebox{.9}{
	\begin{tikzpicture}[auto,semithick
					virtual/.style={thick,densely dashed},
					active/.style={thick,->,shorten >=2pt,shorten <=2pt,>=stealth},
					inactive/.style={color=lightgray,->,shorten >=2pt,shorten <=2pt,>=stealth},
					none/.style={draw=none}
					]
					\tikzstyle{every state}=[rectangle,rounded corners,fill=none,draw=black,text=black,initial text=]
					\tikzstyle{block} = [rectangle,draw=black,text=black,initial text=, minimum size=1cm]
					
					\node[state] 	(A) at (0,0) {$\vee$};
					%\node[block]    [align=left, anchor=west,xshift=.3cm](B) at (A.east) {$r\rightarrow v$};
					
					\draw[inactive] ([xshift=.08cm]A.north) -- (0.08,1.2);
					\draw[active] (-0.08,1.2) -- ([xshift=-.08cm]A.north);
					
					\draw[active] (A.217) -- (-1.05,-.95);
					\draw[inactive] (-.95,-1.05) -- (A.233);
					
					\draw[inactive] (A.307) -- (.95,-1.05);
					\draw[inactive] (1.05,-.95) -- (A.323);
					\end{tikzpicture}
	\begin{tikzpicture}[auto,semithick
					virtual/.style={thick,densely dashed},
					active/.style={thick,->,shorten >=2pt,shorten <=2pt,>=stealth},
					inactive/.style={color=lightgray,->,shorten >=2pt,shorten <=2pt,>=stealth},
					none/.style={draw=none}
					]
					\tikzstyle{every state}=[rectangle,rounded corners,fill=none,draw=black,text=black,initial text=]
					\tikzstyle{block} = [rectangle,draw=black,text=black,initial text=, minimum size=1cm]
					
					\node[state] 	(A) at (0,0) {$\vee$};
					%\node[block]    [align=left, anchor=west,xshift=.3cm](B) at (A.east) {$r\rightarrow v$};
					
					\draw[active] ([xshift=.08cm]A.north) -- (0.08,1.2);
					\draw[inactive] (-0.08,1.2) -- ([xshift=-.08cm]A.north);
					
					\draw[inactive] (A.217) -- (-1.05,-.95);
					\draw[inactive] (-.95,-1.05) -- (A.233);
					
					\draw[inactive] (A.307) -- (.95,-1.05);
					\draw[active] (1.05,-.95) -- (A.323);
					\end{tikzpicture}
	\begin{tikzpicture}[auto,semithick
					virtual/.style={thick,densely dashed},
					active/.style={thick,->,shorten >=2pt,shorten <=2pt,>=stealth},
					inactive/.style={color=lightgray,->,shorten >=2pt,shorten <=2pt,>=stealth},
					none/.style={draw=none}
					]
					\tikzstyle{every state}=[rectangle,rounded corners,fill=none,draw=black,text=black,initial text=]
					\tikzstyle{block} = [rectangle,draw=black,text=black,initial text=, minimum size=.6cm]
					
					\node[state] 	(A) at (0,0) {$\vee$};
					%\node[block]    [align=left, anchor=west,xshift=.3cm](B) at (A.east) {$r\rightarrow v\wedge r$};
					
					\draw[active] ([xshift=.08cm]A.north) -- (0.08,1.2) node[midway, right] {$r=\top$};
					\draw[inactive] (-0.08,1.2) -- ([xshift=-.08cm]A.north);
					
					\draw[inactive] (A.217) -- (-1.05,-.95);
					\draw[active] (-.95,-1.05) -- (A.233);
					
					\draw[active] (A.307) -- (.95,-1.05);
					\draw[inactive] (1.05,-.95) -- (A.323) node[midway, right, color=black] {$r=\bot$};
					\end{tikzpicture}}\\
	\scalebox{.9}{
	\begin{tikzpicture}[auto,semithick
					virtual/.style={thick,densely dashed},
					active/.style={thick,->,shorten >=2pt,shorten <=2pt,>=stealth},
					inactive/.style={color=lightgray,->,shorten >=2pt,shorten <=2pt,>=stealth},
					none/.style={draw=none}
					]
					\tikzstyle{every state}=[rectangle,rounded corners,fill=none,draw=black,text=black,initial text=]
					\tikzstyle{block} = [rectangle,draw=black,text=black,initial text=, minimum size=.6cm]
					
					\node[state] 	(A) at (0,0) {$\neg$};
					%\node[block]    [align=left, anchor=west,xshift=.3cm](B) at (A.east) {$v\rightarrow 1-v$};
					
					\draw[inactive] ([xshift=.08cm]A.north) -- (0.08,1.2);
					\draw[active] (-0.08,1.2) -- ([xshift=-.08cm]A.north);
					
					\draw[active] (A.217) -- (-1.05,-.95);
					\draw[inactive] (-.95,-1.05) -- (A.233);
					
					\draw[none] (A.307) -- (.95,-1.05);
					\draw[none] (1.05,-.95) -- (A.323);
					\end{tikzpicture}
	\begin{tikzpicture}[auto,semithick
					virtual/.style={thick,densely dashed},
					active/.style={thick,->,shorten >=2pt,shorten <=2pt,>=stealth},
					inactive/.style={color=lightgray,->,shorten >=2pt,shorten <=2pt,>=stealth},
					none/.style={draw=none}
					]
					\tikzstyle{every state}=[rectangle,rounded corners,fill=none,draw=black,text=black,initial text=]
					\tikzstyle{block} = [rectangle,draw=black,text=black,initial text=, minimum size=.6cm]
					
					\node[state] 	(A) at (0,0) {$\neg$};
					\node[block]    [align=left, anchor=west,xshift=.3cm](B) at (A.east) {$r\rightarrow \overline{r}$};
					
					\draw[active] ([xshift=.08cm]A.north) -- (0.08,1.2);
					\draw[inactive] (-0.08,1.2) -- ([xshift=-.08cm]A.north);
					
					\draw[inactive] (A.217) -- (-1.05,-.95);
					\draw[active] (-.95,-1.05) -- (A.233);
					
					\draw[none] (A.307) -- (.95,-1.05);
					\draw[none] (1.05,-.95) -- (A.323);
					\end{tikzpicture}}\\
	\scalebox{.9}{
	\begin{tikzpicture}[auto,semithick
					virtual/.style={thick,densely dashed},
					active/.style={thick,->,shorten >=2pt,shorten <=2pt,>=stealth},
					inactive/.style={color=lightgray,->,shorten >=2pt,shorten <=2pt,>=stealth},
					none/.style={draw=none}
					]
					\tikzstyle{every state}=[rectangle,rounded corners,fill=none,draw=black,text=black,initial text=]
					\tikzstyle{block} = [rectangle,draw=black,text=black,initial text=, minimum size=.6cm]
					
					\node[state] 	(A) at (0,0) {\textit{while}};
					%\node[block]    [align=left, anchor=west,xshift=.3cm](B) at (A.east) {$\{ v\rightarrow b \mid b\in\mb{B} \}$\IGNORE{$\forall b\in\mb{B}: v\rightarrow b$}};
					
					\draw[inactive] ([xshift=.08cm]A.north) -- (0.08,1.2);
					\draw[active] (-0.08,1.2) -- ([xshift=-.08cm]A.north);
					
					\draw[active] (A.217) -- (-1.05,-.95);
					\draw[inactive] (-.95,-1.05) -- (A.233);
					
					\draw[inactive] (A.307) -- (.95,-1.05);
					\draw[active] (1.05,-.95) -- (A.323);
					\end{tikzpicture}
	\begin{tikzpicture}[auto,semithick
					virtual/.style={thick,densely dashed},
					active/.style={thick,->,shorten >=2pt,shorten <=2pt,>=stealth},
					inactive/.style={color=lightgray,->,shorten >=2pt,shorten <=2pt,>=stealth},
					none/.style={draw=none}
					]
					\tikzstyle{every state}=[rectangle,rounded corners,fill=none,draw=black,text=black,initial text=]
					\tikzstyle{block} = [rectangle,draw=black,text=black,initial text=, minimum size=.6cm]
					
					\node[state] 	(A) at (0,0) {\textit{while}};
					%\node[block]    [align=left, anchor=west,xshift=.3cm](B) at (A.east) {$r\rightarrow v\wedge r$};
					
					\draw[active] ([xshift=.08cm]A.north) -- (0.08,1.2) node[midway, right] {$r=\bot$};
					\draw[inactive] (-0.08,1.2) -- ([xshift=-.08cm]A.north);
					
					\draw[inactive] (A.217) -- (-1.05,-.95);
					\draw[active] (-.95,-1.05) -- (A.233);
					
					\draw[active] (A.307) -- (.95,-1.05);
					\draw[inactive] (1.05,-.95) -- (A.323) node[midway, right, color=black] {$r=\top$};
					\end{tikzpicture}
	\hskip 20pt
	\begin{tikzpicture}[auto,semithick
					virtual/.style={thick,densely dashed},
					active/.style={thick,->,shorten >=2pt,shorten <=2pt,>=stealth},
					inactive/.style={color=lightgray,->,shorten >=2pt,shorten <=2pt,>=stealth},
					none/.style={draw=none}
					]
					\tikzstyle{every state}=[rectangle,rounded corners,fill=none,draw=black,text=black,initial text=]
					\tikzstyle{block} = [rectangle,draw=black,text=black,initial text=, minimum size=.6cm]
					
					\node[state] 	(A) at (0,0) {\textit{assign}$_b$};
					%\node[block]    [align=left, anchor=west,xshift=.3cm](B) at (A.east) {$\{ v\rightarrow b \mid b\in\mb{B} \}$\IGNORE{$\forall b\in\mb{B}: v\rightarrow b$}};
					
					\draw[inactive] ([xshift=.08cm]A.north) -- (0.08,1.2);
					\draw[active] (-0.08,1.2) -- ([xshift=-.08cm]A.north);
					
					\draw[active] (A.217) -- (-1.05,-.95);
					\draw[inactive] (-.95,-1.05) -- (A.233);
					
					\draw[none] (A.307) -- (.95,-1.05);
					\draw[none] (1.05,-.95) -- (A.323);
					\end{tikzpicture}
	\begin{tikzpicture}[auto,semithick
					virtual/.style={thick,densely dashed},
					active/.style={thick,->,shorten >=2pt,shorten <=2pt,>=stealth},
					inactive/.style={color=lightgray,->,shorten >=2pt,shorten <=2pt,>=stealth},
					none/.style={draw=none}
					]
					\tikzstyle{every state}=[rectangle,rounded corners,fill=none,draw=black,text=black,initial text=]
					\tikzstyle{block} = [rectangle,draw=black,text=black,initial text=, minimum size=.6cm]
					
					\node[state] 	(A) at (0,0) {\textit{assign}$_b$};
					\node[block]    [align=left, anchor=west,xshift=.3cm](B) at (A.east) {$s[b]\rightarrow r$};
					
					\draw[active] ([xshift=.08cm]A.north) -- (0.08,1.2);
					\draw[inactive] (-0.08,1.2) -- ([xshift=-.08cm]A.north);
					
					\draw[inactive] (A.217) -- (-1.05,-.95);
					\draw[active] (-.95,-1.05) -- (A.233);
					
					\draw[none] (A.307) -- (.95,-1.05);
					\draw[none] (1.05,-.95) -- (A.323);
					\end{tikzpicture}}\\
	\scalebox{.9}{
	\begin{tikzpicture}[auto,semithick
virtual/.style={thick,densely dashed},
active/.style={thick,->,shorten >=2pt,shorten <=2pt,>=stealth},
inactive/.style={color=lightgray,->,shorten >=2pt,shorten <=2pt,>=stealth},
none/.style={draw=none}
]
\tikzstyle{every state}=[rectangle,rounded corners,fill=none,draw=black,text=black,initial text=]
\tikzstyle{block} = [rectangle,draw=black,text=black,initial text=, minimum size=.6cm]

\node[state] 	(A) at (0,0) {\textit{if}};
\node[state,color=lightgray]	(B) at (1.5,-1.5) {\textit{then}};
%\node[block]    [align=left, anchor=west,xshift=.3cm](C) at (A.east) {$r\rightarrow v\wedge r$};

\draw[inactive] ([xshift=.08cm]A.north) -- (0.08,1.2);
\draw[active] (-0.08,1.2) -- ([xshift=-.08cm]A.north);

\draw[active] (A.217) -- (-1.05,-.95);
\draw[inactive] (-.95,-1.05) -- (A.233);

%\draw[active] (A.307) -- (.95,-1.05);
%\draw[inactive] (1.05,-.95) -- (A.323) node[midway, right, color=black] {$r\not=v$};

\draw[inactive] (A.307) -- (B.143);
\draw[inactive] (B.127) -- (A.323);

\draw[inactive] (B.217) -- (.45,-2.45);
\draw[inactive] (.55,-2.55) -- (B.233);

\draw[inactive] (B.307) -- (2.45,-2.55);
\draw[inactive] (2.55,-2.45) -- (B.323);
\end{tikzpicture}
	\begin{tikzpicture}[auto,semithick
					virtual/.style={thick,densely dashed},
					active/.style={thick,->,shorten >=2pt,shorten <=2pt,>=stealth},
					inactive/.style={color=lightgray,->,shorten >=2pt,shorten <=2pt,>=stealth},
					none/.style={draw=none}
					]
					\tikzstyle{every state}=[rectangle,rounded corners,fill=none,draw=black,text=black,initial text=]
					\tikzstyle{block} = [rectangle,draw=black,text=black,initial text=, minimum size=.6cm]
					
					\node[state] 	(A) at (0,0) {\textit{if}};
					\node[state]	(B) at (1.5,-1.5) {\textit{then}};
					%\node[block]    [align=left, anchor=west,xshift=.3cm](C) at (A.east) {$r\rightarrow v\wedge r$};
					
					\draw[inactive] ([xshift=.08cm]A.north) -- (0.08,1.2);
					\draw[inactive] (-0.08,1.2) -- ([xshift=-.08cm]A.north);
					
					\draw[inactive] (A.217) -- (-1.05,-.95);
					\draw[active] (-.95,-1.05) -- (A.233);
					
					%\draw[active] (A.307) -- (.95,-1.05);
					%\draw[inactive] (1.05,-.95) -- (A.323) node[midway, right, color=black] {$r\not=v$};
					
					\draw[active] (A.307) -- (B.143);
					\draw[inactive] (B.127) -- (A.323);
					
					\draw[active] (B.217) -- (.45,-2.45) node[midway, left, color=black] {$r=\top$};
					\draw[inactive] (.55,-2.55) -- (B.233);
					
					\draw[active] (B.307) -- (2.45,-2.55);
					\draw[inactive] (2.55,-2.45) -- (B.323) node[midway, right, color=black] {$r=\bot$};
					\end{tikzpicture}
	\begin{tikzpicture}[auto,semithick
					virtual/.style={thick,densely dashed},
					active/.style={thick,->,shorten >=2pt,shorten <=2pt,>=stealth},
					inactive/.style={color=lightgray,->,shorten >=2pt,shorten <=2pt,>=stealth},
					none/.style={draw=none}
					]
					\tikzstyle{every state}=[rectangle,rounded corners,fill=none,draw=black,text=black,initial text=]
					\tikzstyle{block} = [rectangle,draw=black,text=black,initial text=, minimum size=.6cm]
					
					\node[state] 	(A) at (0,0) {\textit{if}};
					\node[state]	(B) at (1.5,-1.5) {\textit{then}};
					%\node[block]    [align=left, anchor=west,xshift=.3cm](C) at (A.east) {$r\rightarrow v\wedge r$};
					
					\draw[active] ([xshift=.08cm]A.north) -- (0.08,1.2);
					\draw[inactive] (-0.08,1.2) -- ([xshift=-.08cm]A.north);
					
					\draw[inactive] (A.217) -- (-1.05,-.95);
					\draw[inactive] (-.95,-1.05) -- (A.233);
					
					%\draw[active] (A.307) -- (.95,-1.05);
					%\draw[inactive] (1.05,-.95) -- (A.323) node[midway, right, color=black] {$r\not=v$};
					
					\draw[inactive] (A.307) -- (B.143);
					\draw[active] (B.127) -- (A.323);
					
					\draw[inactive] (B.217) -- (.45,-2.45);
					\draw[active] (.55,-2.55) -- (B.233);
					
					\draw[inactive] (B.307) -- (2.45,-2.55);
					\draw[active] (2.55,-2.45) -- (B.323);
					\end{tikzpicture}}\\
	\begin{tikzpicture}[auto,semithick
					virtual/.style={thick,densely dashed},
					active/.style={thick,->,shorten >=2pt,shorten <=2pt,>=stealth},
					inactive/.style={color=lightgray,->,shorten >=2pt,shorten <=2pt,>=stealth},
					none/.style={draw=none}
					]
					\tikzstyle{every state}=[rectangle,rounded corners,fill=none,draw=black,text=black,initial text=]
					\tikzstyle{block} = [rectangle,draw=black,text=black,initial text=, minimum size=.6cm]
					
					\node[state] 	(A) at (0,0) {\textbf{input} $\vec{b}$};
					\node[block]    [align=left, anchor=west,xshift=.3cm](B) at (A.east) {$\{s[\vec{b}]\rightarrow \vec{val}, i\rightarrow \vec{val}, m\rightarrow out\mid \vec{val}\in\mb{B}^{N_\I} \}$};
					
					\draw[active] ([xshift=.08cm]A.north) -- (0.08,1.2);
					\draw[active] (-0.08,1.2) -- ([xshift=-.08cm]A.north) node[midway, left] {$m = inp$};
					
					\draw[none] (A.217) -- (-1.05,-.95);
					\draw[none] (-.95,-1.05) -- (A.233);
					
					\draw[none] (A.307) -- (.95,-1.05);
					\draw[none] (1.05,-.95) -- (A.323);
					\end{tikzpicture}\\
	\begin{tikzpicture}[auto,semithick
					virtual/.style={thick,densely dashed},
					active/.style={thick,->,shorten >=2pt,shorten <=2pt,>=stealth},
					inactive/.style={color=lightgray,->,shorten >=2pt,shorten <=2pt,>=stealth},
					none/.style={draw=none}
					]
					\tikzstyle{every state}=[rectangle,rounded corners,fill=none,draw=black,text=black,initial text=]
					\tikzstyle{block} = [rectangle,draw=black,text=black,initial text=, minimum size=.6cm]
					
					\node[state] 	(A) at (0,0) {\textbf{output} $\vec{b}$};
					\node[block]    [align=left, anchor=west,xshift=.3cm](B) at (A.east) {$\{q \rightarrow q', m\rightarrow inp, t \rightarrow 1\mid q'\in\delta(q,i,s[\vec{b}])\}$};
					
					\draw[active] ([xshift=.08cm]A.north) -- (0.08,1.2);
					\draw[active] (-0.08,1.2) -- ([xshift=-.08cm]A.north) node[midway, left] {$m = out$};
					
					\draw[none] (A.217) -- (-1.05,-.95);
					\draw[none] (-.95,-1.05) -- (A.233);
					
					\draw[none] (A.307) -- (.95,-1.05);
					\draw[none] (1.05,-.95) -- (A.323);
					\end{tikzpicture}\\
	\scalebox{.9}{
	\begin{tikzpicture}[auto,semithick
					virtual/.style={thick,densely dashed},
					active/.style={thick,->,shorten >=2pt,shorten <=2pt,>=stealth},
					inactive/.style={color=lightgray,->,shorten >=2pt,shorten <=2pt,>=stealth},
					none/.style={draw=none}
					]
					\tikzstyle{every state}=[rectangle,rounded corners,fill=none,draw=black,text=black,initial text=]
					\tikzstyle{block} = [rectangle,draw=black,text=black,initial text=, minimum size=.6cm]
					
					\node[state] 	(A) at (0,0) {\textbf{skip}};
					%\node[block]    [align=left, anchor=west,xshift=.3cm](B) at (A.east) {$r\rightarrow v\wedge r$};
					
					\draw[active] ([xshift=.08cm]A.north) -- (0.08,1.2);
					\draw[active] (-0.08,1.2) -- ([xshift=-.08cm]A.north);
					
					\draw[none] (A.217) -- (-1.05,-.95);
					\draw[none] (-.95,-1.05) -- (A.233);
					
					\draw[none] (A.307) -- (.95,-1.05);
					\draw[none] (1.05,-.95) -- (A.323);
					\end{tikzpicture}
	\begin{tikzpicture}[auto,semithick
					virtual/.style={thick,densely dashed},
					active/.style={thick,->,shorten >=2pt,shorten <=2pt,>=stealth},
					inactive/.style={color=lightgray,->,shorten >=2pt,shorten <=2pt,>=stealth},
					none/.style={draw=none}
					]
					\tikzstyle{every state}=[rectangle,rounded corners,fill=none,draw=black,text=black,initial text=]
					\tikzstyle{block} = [rectangle,draw=black,text=black,initial text=, minimum size=1cm]
					
					\node[state] 	(A) at (0,0) {\textbf{;}};
					%\node[block]    [align=left, anchor=west,xshift=.5cm](B) at (A.east) {/};
					
					\draw[inactive] ([xshift=.08cm]A.north) -- (0.08,1.2);
					\draw[active] (-0.08,1.2) -- ([xshift=-.08cm]A.north);
					
					\draw[active] (A.217) -- (-1.05,-.95);
					\draw[inactive] (-.95,-1.05) -- (A.233);
					
					\draw[inactive] (A.307) -- (.95,-1.05);
					\draw[inactive] (1.05,-.95) -- (A.323);
					\end{tikzpicture}
	\begin{tikzpicture}[auto,semithick
					virtual/.style={thick,densely dashed},
					active/.style={thick,->,shorten >=2pt,shorten <=2pt,>=stealth},
					inactive/.style={color=lightgray,->,shorten >=2pt,shorten <=2pt,>=stealth},
					none/.style={draw=none}
					]
					\tikzstyle{every state}=[rectangle,rounded corners,fill=none,draw=black,text=black,initial text=]
					\tikzstyle{block} = [rectangle,draw=black,text=black,initial text=, minimum size=1cm]
					
					\node[state] 	(A) at (0,0) {\textbf{;}};
					%\node[block]    [align=left, anchor=west,xshift=.5cm](B) at (A.east) {/};
					
					\draw[inactive] ([xshift=.08cm]A.north) -- (0.08,1.2);
					\draw[inactive] (-0.08,1.2) -- ([xshift=-.08cm]A.north);
					
					\draw[inactive] (A.217) -- (-1.05,-.95);
					\draw[active] (-.95,-1.05) -- (A.233);
					
					\draw[active] (A.307) -- (.95,-1.05);
					\draw[inactive] (1.05,-.95) -- (A.323);
					\end{tikzpicture}
	\begin{tikzpicture}[auto,semithick
					virtual/.style={thick,densely dashed},
					active/.style={thick,->,shorten >=2pt,shorten <=2pt,>=stealth},
					inactive/.style={color=lightgray,->,shorten >=2pt,shorten <=2pt,>=stealth},
					none/.style={draw=none}
					]
					\tikzstyle{every state}=[rectangle,rounded corners,fill=none,draw=black,text=black,initial text=]
					\tikzstyle{block} = [rectangle,draw=black,text=black,initial text=, minimum size=1cm]
					
					\node[state] 	(A) at (0,0) {\textbf{;}};
					%\node[block]    [align=left, anchor=west,xshift=.5cm](B) at (A.east) {/};
					
					\draw[active] ([xshift=.08cm]A.north) -- (0.08,1.2);
					\draw[inactive] (-0.08,1.2) -- ([xshift=-.08cm]A.north);
					
					\draw[inactive] (A.217) -- (-1.05,-.95);
					\draw[inactive] (-.95,-1.05) -- (A.233);
					
					\draw[inactive] (A.307) -- (.95,-1.05);
					\draw[active] (1.05,-.95) -- (A.323);
					\end{tikzpicture}}
	
	\caption{Semantics of the constructed two-way automata}\label{semantics}
\end{figure}
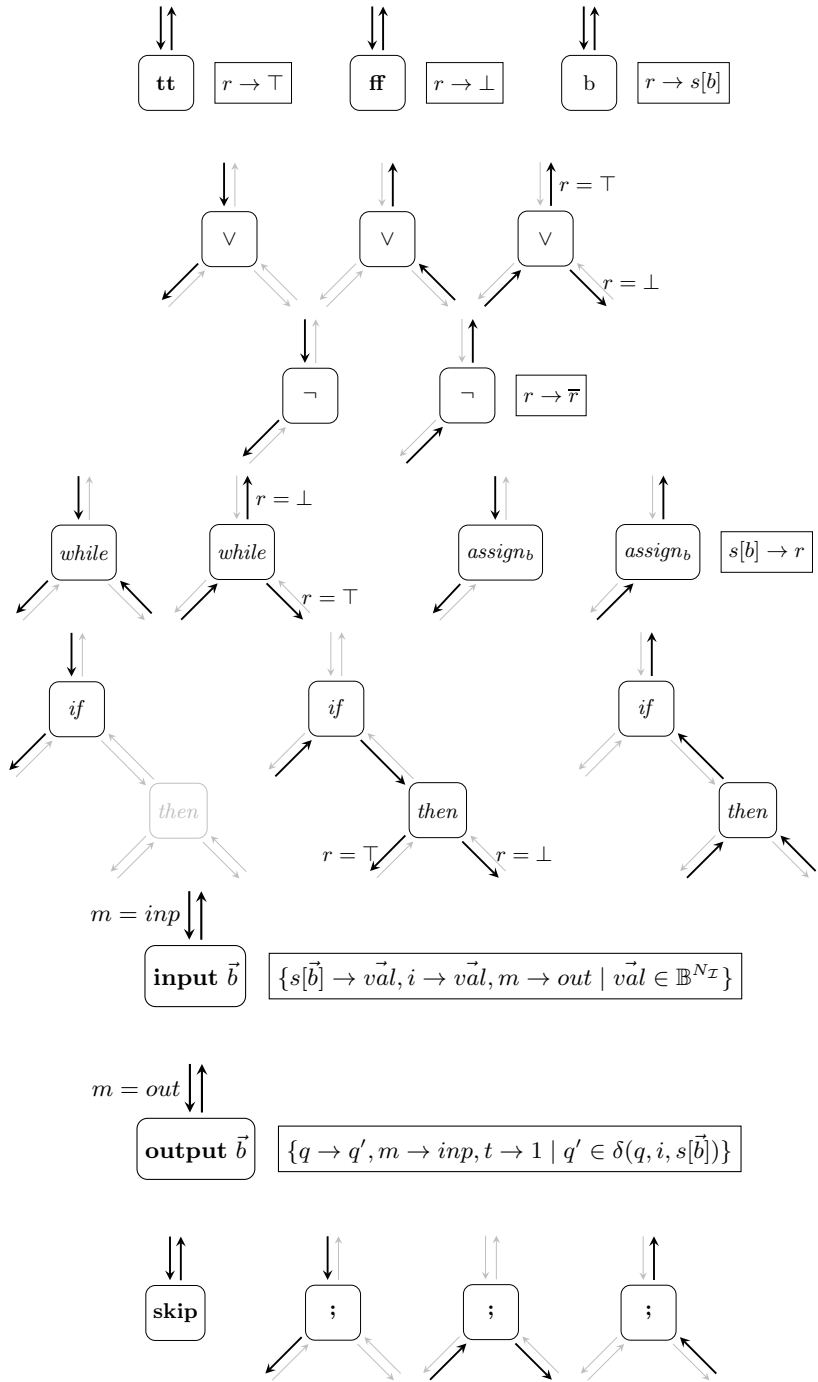

The set of accepting states is defined as $$F^\B = \big\{ (s,q,i,m,1) \mid q \in F_{spec} \big\}$$

A formal construction of $\B$ is given in \cref{explicitB}. Note that $\B$ behaves similar to $\A$ during normal execution and that only Boolean evaluation was altered. Therefore, the state spaces of the automata only differ in the states corresponding to Boolean evaluation and especially the sets of accepting states $F^\A$ and $F^\B$ are equivalent. Therefore, we can prove the equivalence by showing that both automata visit the same sequences of accepting states and thus accept the same program trees. 

\begin{theorem}[\cref{appendixEquivalenceProof}]
	$\mc{L}(\A) = \mc{L}(\B)$
\end{theorem}

We now complement the constructed two-way non-deterministic Büchi automaton into a two-way universal co-Büchi automaton. From this point onwards, we refer with $\B$ to the two-way universal co-Büchi automaton. 

Since $\A$ accepts precisely the programs that fail the specification and interact infinitely often with the environment, the complement now only accepts programs that do satisfy the specification or interact finitely often with the environment. We fix the remaining misbehavior by enforcing syntactical correctness and reactiveness.

\subsection{Guarantee Syntactical Correctness}

Due to the fact that $\B$ was designed to correctly simulate programs of our defined syntax and transitions were only defined for syntax-valid statements, $\B$ implicitly rejects programs that are syntactically invalid. But such programs are only then rejected when their syntactically incorrect statements are traversed in the simulation, therefore $\B$ does not check for syntactically correct subtrees that are unreachable. It is now arguable whether the syntax check is necessary in practice. One could expect programs to be syntactically correct in total and this expectation is in general well-argued. On the other hand, we do perform bounded synthesis, i.e., we search for implementations with a bound on the implementation size and then increment this bound until a valid implementation is found. It is easy to see that programs with unreachable parts can be represented by smaller programs with the same behavior simply by removing unreachable statements. Therefore, with an incremental search one first finds the smallest and thus syntactically correct programs.

\subsection{Guarantee Reactiveness}
It now remains to guarantee reactiveness of the programs accepted by $\B$. For that purpose, we introduce a two-way universal Büchi automaton $\B_{reactive}$, which only accepts program trees that are reactive. This automaton is designed with the exact same states and transitions as $\B$ but with another acceptance condition. The intersection of $\B$ and $\B_{reactive}$ then yields a two-way universal Streett automaton $\B'$. We construct $\B_{reactive}$ with the set of accepting states:
$$F^\B_{reactive} = \big\{ (s,q,i,m,1) \mid \forall s,q,i,m \big\}$$
$\B_{reactive}$ accepts a program tree, iff it produces infinitely many outputs on all possible executions. Due to the alternation between input and output statements the program reacts infinitely often with its environment, i.e., it is reactive.

Formally, $\B'$ is the tuple $(\Sigma_P, P^\B, \delta^\B_L, \delta^\B_R, \delta^\B_{LR}, \delta^\B_\emptyset, \text{STREETT}(F^{\B'}))$, where $$F^{\B'}=\big\{ (F^\B, \emptyset), (P^\B, F^\B_{reactive}) \big\}.$$

\begin{lemma}
	$\mc{L}(\B') = \mc{L}(\B)\cap\mc{L}(\B_{reactive})$
\end{lemma}
\begin{proof}
	Besides the acceptance condition, all three automata are equivalent. The tuples of the Streett condition $(F^\B,\emptyset)$ and $(P^\B, F^\B_{reactive})$ express the co-Büchi and Büchi condition of $\B$ and $\B_{reactive}$, respectively. \qed
\end{proof}

We capture the complete construction by the following theorem.

\begin{theorem}
	Let $B$	be a finite set of Boolean variables and $\varphi$ a specification given as LTL-formula. The constructed two-way universal Streett automaton~$\B'$ accepts program trees over $B$ that satisfy the specification.
\end{theorem}

\section{Bounded Synthesis}
In this section, we generalize the bounded synthesis approach towards arbitrary universal automata and then apply it to the constructed two-way automaton to synthesize bounded programs.

We fix $Q$ to be a finite set of states. A run graph is a tuple $\G = (V, v_0, E, f)$, where $V$ is a finite set of vertices, $v_0$ is an initial vertex, $E\subseteq V\times V$ is a set of directed edges and $f: V \rightarrow Q$ is a labeling function. A path \mbox{$\pi=\pi_0\pi_1\hdots\in V^\omega$} is \textit{contained} in $\G$, denoted by $\pi\in\G$, iff $\forall i\in\mb{N}: (\pi_i,\pi_{i+1})\in E$ and $\pi_0 = v_0$, i.e., a path in the graph starting in the initial vertex.
We denote with $f(\pi)=f(\pi_0)f(\pi_1)\hdots\in Q^\omega$ the application of $f$ on every node in the path, i.e., a projection to an infinite sequence of states.
We call a vertex $v$ \textit{unreachable}, iff there exists no path $\pi\in\G$ containing $v$. Let $Acc\subseteq Q^\omega$ be an acceptance condition. We say $\G$ \textit{satisfies} $Acc$, iff every path of $\G$ satisfies the acceptance condition, i.e., $\forall\pi\in\G: f(\pi)\in Acc$.

Run graphs are used to express all possible runs of a universal automaton on some implementation. This is usually done for universal word automata on Mealy machines, but we need a generalized version to later utilize it for two-way universal tree automata on program trees. Let $\Sigma=2^{\I\cup\O}$. We define a run graph $\G^A_\M=(V,v_0,E,f)$ of a universal word automaton $A=(\Sigma,Q,q_0,\delta,Acc)$ on a Mealy machine $\M=(\I,\O,M,m_0,\tau,o)$ as an instantiation of the given definition, where 
\begin{itemize}
	\item $V=Q\times M$, 
	\item $v_0=(q_0,m_0)$, 
	\item $E=\big\{ \big((q,m),(q',m')\big)\mid \exists in\in2^\I,out\in2^\O:$\\
	$\tau(m,in)=m' \;\wedge\; o(m,in)=out \;\wedge\; q'\in\delta(q,in\cup out) \big\}$ and 
	\item $f(q,m)=q$.
\end{itemize}

Since the run graph contains all infinite runs of $A$ on words producible by $\M$, $A$ \textit{accepts} $\M$, iff all runs in $\G^A_\M$ are accepting, i.e., $\G^A_\M$ satisfies $Acc$.

For some bound $c\in \mb{N}$ we denote $\{0,1,\hdots,c\}$ by $D_c$.
For a run graph $\G=(V,v_0,E,f)$ and a bound $c\in\mb{N}$ a $c$-bounded annotation function on $\G$ is a function $\lambda:V\rightarrow D_c$.
An annotation comparison relation of arity $n$ is a family of relations $\triangleright = (\triangleright_0,\triangleright_1,\hdots, \triangleright_{n-1}) \in (2^{Q\times D_c \times D_c})^n$. We refer to $\triangleright_i\subseteq Q\times D_c \times D_c$ as basic comparison relations for $i\in[n]$. We denote the arity with $|\triangleright|=n$. We write $\lambda(v)\triangleright_i\lambda(v')$ for $(f(v),\lambda(v),\lambda(v'))\in\triangleright_i$ and for comparison relations of arity $|\triangleright|=1$ we omit the index.

We say a path $\pi\in\G$ \textit{satisfies} a comparison relation $\triangleright$ with arity $|\triangleright|=n$, denoted by $\pi\models\triangleright$, iff for every basic comparison relation there exists an annotation function that annotates every node with a value such that the annotated number for all consecutive nodes in the path satisfy the basic comparison relation, i.e., $\forall i\in[n]:\exists\lambda:\forall j\in\mb{N}:\lambda(\pi_j)\triangleright_i\lambda(\pi_{j+1}).$ For an acceptance condition $Acc\subseteq Q^\omega$ we say a comparison relation $\triangleright$ \textit{expresses} $Acc$, iff all paths in $\G$ satisfy the relation if and only if the path satisfies the acceptance condition, i.e., $\forall\pi\in\G: \pi\models\triangleright \leftrightarrow f(\pi)\in Acc.$ A $c$-bounded annotation function $\lambda$ on $\G=(V,v_0,E,f)$ is \textit{valid} for a basic annotation comparison relation $\triangleright\subseteq Q\times D_c\times D_c$, iff $\text{for all reachable } v,v'\in V: (v,v')\in E \rightarrow \lambda(v)\triangleright\lambda(v').$

We use the following annotation comparison relations to express Büchi, co-Büchi and Streett acceptance conditions.		
	\begin{itemize}
		\item Let $F\subseteq Q$ and $Acc = \text{BÜCHI}(F)$. Then $\;\triangleright^F_B$ is defined as
		$$\lambda(v)\triangleright^F_B\lambda(v') = \begin{cases}
		true & \text{if } f(v)\in F\\
		\lambda(v)>\lambda(v') & \text{if } f(v)\not\in F
		\end{cases}$$
		\item Let $F\subseteq Q$ and $Acc = \text{CO-BÜCHI}(F)$. Then $\;\triangleright^F_C$ is defined as
		$$\lambda(v)\triangleright^F_C\lambda(v') = \begin{cases}
		\lambda(v)>\lambda(v') & \text{if } f(v)\in F\\
		\lambda(v)\geq\lambda(v') & \text{if } f(v)\not\in F
		\end{cases}$$
		\item Let $F=\{(A_i,G_i)\}_{i\in[k]}\subseteq 2^{Q\times Q}$ and $Acc = \text{STREETT}(F)$. Then $\;\triangleright^F_S = (\triangleright^{F,0}_S, \triangleright^{F,1}_S, \hdots, \triangleright^{F,k-1}_S)$ is defined as
		$$\lambda(v)\triangleright^{F,i}_S\lambda(v') = \begin{cases}
		true & \text{if } f(v)\in G_i\\
		\lambda(v)>\lambda(v') & \text{if } f(v)\in A_i \wedge f(v)\not\in G_i\\
		\lambda(v)\geq\lambda(v') & \text{if } f(v)\not\in A_i\cup G_i
		\end{cases}$$
	\end{itemize}
	
	Note that $|\triangleright^F_B|=|\triangleright^F_C|=1$ and $|\triangleright^F_S|=k$.
	
	\begin{theorem}[\cite{2013:Schewe:BS,2017:Khalimov:BS}]\label{bsold}
		Let $F$ be a set, the acceptance condition of $A$ be expressed by $\triangleright^F_X$ with $X\in\{B,C,S\}$, $c\in\mb{N}$ a bound and $\G^A_\M$ the run graph of $A$ on $\M$.
		
		If and only if, there exists a valid $c$-bounded annotation function $\lambda_i$ on $\G^A_\M$ for each basic comparison relation $\triangleright_i$, then $\G^A_\M$ satisfies $Acc$.
	\end{theorem}

\subsection{General Bounded Synthesis}
\label{general}
In \cref{bsold} we saw that the acceptance of a Mealy machine $\M$ by a universal automata $A$ can be expressed by the existence of an annotation comparison relation. To do the same for two-way automata on program trees, we generalize this theorem towards arbitrary run graphs.

Let $\A=(\Sigma_P,P,p_0, \delta_L, \delta_R, \delta_{LR}, \delta_\emptyset, Acc)$ be a two-way universal tree automaton and $\T=(T,\tau)$ a program tree. We define the run graph of $\A$ on $\T$ as $\G^\A_\T=(V,v_0,E,f)$, where
\begin{itemize}
	\item $V=P\times T\times \{L,R,D \}$,
	\item $v_0=(p_0,\epsilon,D)$,
	\item $E=\big\{ \big((p,t,d),(p',t',d')\big)\mid \exists d''\in\{ L,R,U \}: \mu(t,d'')=(t',d')$\\$\wedge (p',d'')\in\delta_t(p,\tau(t),d) \big\}$ and 
	\item $f(p,t,d)=p$.
\end{itemize}

For the generalized encoding, we use the same construction for the annotation comparison relation as presented in \cite{2017:Khalimov:BS} for Street acceptance conditions, which conveniently suffices for the general run graphs. Büchi and co-Büchi then follow as special cases.

\begin{lemma}[\cref{generalStrettProof}]\label{generalStreett}
	For a Streett acceptance condition $Acc=\text{Streett}(F)$ with set of tuples of states $F\subseteq 2^{Q\times Q}$ and a run graph $\G=(V,v_0,E,f)$:
	
	If $\G$ satisfies $Acc$, then there exists a valid $|V|$-bounded annotation function~$\lambda$ for each basic comparison relation in $\triangleright^F_S$.
\end{lemma}

\begin{theorem}\label{BST}
	Let $\G=(V,v_0,E,f)$ be a run graph, $Acc \subseteq Q^\omega$ a Büchi, co-Büchi or Streett acceptance condition expressed by the relation $\triangleright_X$ for $X\in\{B,C,S\}$.
	
	There exists a valid $|V|$-bounded annotation function $\lambda_i$ on $\G$ for each basic comparison relation $\triangleright_i$, if and only if $\G$ satisfies $Acc$.
\end{theorem}
\begin{proof}
	$"\Implies":$ Let $\G$, $Acc$, $\triangleright$ with arity $|\triangleright|=n$ and $c$ be given and $\lambda_i$ be a valid $c$-bounded annotation comparison relation on $\G$ for $\triangleright_i$ for all $i\in[n]$. Let $\pi=\pi_0\pi_1\hdots\in\G$ be an arbitrary path in $\G$ and $i\in[n]$. Since $\lambda_i$ is a valid annotation function, $\lambda_i(\pi_0)\triangleright_i\lambda_i(\pi_1)\triangleright_i\hdots$ holds and therefore $\pi\models\triangleright$. Since $\triangleright$ expresses $Acc$ it follows that $f(\pi)\in Acc$, i.e., $\G$ satisfies $Acc$.\\
	$"\Leftarrow":$ \cref{generalStreett}. \qed
\end{proof}

\subsection{General Encoding}

We showed that the run graph satisfies an acceptance condition $Acc$, iff the implementation is accepted by the automaton. We also proved that the satisfaction of $Acc$ by a run graph can be expressed by the existence of valid annotation functions.

We encode these constraints in SAT. The valid implementation can then be extracted from the satisfied encoding.
Note that in our definition of program trees the structure was implicitly expressed by the nodes and for the encoding we need to express them explicitly. Therefore, the structure of the tree is encoded with successor functions $L$ and $R$, expressing the left and right child of a node, respectively. We encode the program tree and the annotation function as uninterpreted functions as explained in the following. We introduce the following variables for arbitrary two-way automata~$\A$, program trees $\T$, bounds $c$ and annotation comparison relations $\triangleright$:
\begin{itemize}
	\item $\tau_t$ encodes label $l$ of $t$ with $\log(|\Sigma|)$ many variables, notated as $\tau_t \equiv l$
	\item $L_t$ iff $t$ has left child (implicitly the next program state $t+1$)
	\item $R_t$ encodes right the child of $t\in T$ with $\log(|T|)$ many variables
	\item $\lambda_{p,t,d}^\mb{B}$ iff state $(p,t,d)$ is reachable in the run graph
	\item $\lambda^\#_{i,p,t,d}$ encodes the $i$-th annotation of state $(p,t,d)$ with $\log(c)$ many variables. We omit the index $i$ in the encoding
\end{itemize}	

The SAT formula $\Phi^{\A,\triangleright}_\T$ consists of the following constraints:

\begin{itemize}
	\item The initial state is reachable and all annotations fulfill the given bound:
	$$\lambda^\mb{B}_{p_0,t_0,D} \wedge \bigwedge_{\substack{p\in P,\\t\in T,\\d\in\{L,R,D\}}} \lambda^\#_{p,t,d} \leq c$$
	\item Bounded synthesis encoding
	$$\bigwedge_{\substack{p\in P,\\ t\in T,\\ d\in D}} \lambda_{p,t,d}^\mb{B} \rightarrow \bigwedge_{\substack{\sigma\in \Sigma}} (\tau_t \equiv \sigma) \rightarrow \bigwedge_{\substack{(p',d'')\in\delta(p,\sigma,d),\\ t'\in T, \\ (\varphi, d')\in \mu'(t,d'',t')}} \varphi\; \rightarrow \; \lambda_{p',t',d'}^\mb{B} \wedge \lambda_{p,t,d}^\# \triangleright \lambda_{p',t',d'}^\#$$
	$\mu' : T \times D' \times T \rightarrow [\mb{B}(L_t,R_t) \times D]$ returns a list of pairs $(\varphi, d')$, where the formula $\varphi$ enforces the tree structure needed to reach $p',t',d'$.
\end{itemize}

The encoding checks whether universal properties in the run graph hold. Note that we need to additionally forbid walking up from the root node, which is omitted here.

\begin{theorem}
	Given a two-way universal tree automaton $\A$ with a Büchi, co-Büchi or Streett acceptance condition $Acc$ expressed by $\triangleright$ and a bound $c\in \mb{N}$. The constraint system $\Phi^{\A,\triangleright}_\T$ is satisfiable, iff there is a program tree $\T$ with size $|\T|\leq\lfloor c/|\A|\rfloor$ that is accepted by $\A$.
\end{theorem}
\begin{proof}
  $"\Implies":$ Let $\T$ be accepted by $\A$, then with \cref{BST} there exists a valid annotation function $\lambda_i$ on $\G$ for each $i\in[|\triangleright|]$. Let $\lambda_i$ be represented by $\lambda_i^\#$ and $\lambda_i^\mb{B}$ be $true$ for all reachable states in the run graph $\G$. Then $\Phi^{\A,\triangleright}_\T$ is satisfied.
	
  $"\Leftarrow":$ Let $\Phi^{\A,\triangleright}_\T$ be satisfied. Then there exists a valid annotation function $\lambda_i$ encoded by $\lambda_i^\#$ for each $i\in[|\triangleright|]$ (set $\lambda_i(v)=0$ for all unreachable states $v$, i.e., where $\lambda_i^\#(v)$ is $false$) that satisfies the encoding. With \cref{BST} the acceptance of $\T$ by $\A$ follows. \qed
\end{proof}

Utilizing this theorem, we now can by means of the encoding $\Phi^{\B'}_S$ synthesize program trees accepted by $\B'$, i.e., precisely those program trees, which correspond to reactive programs that satisfy the given specification the automaton was constructed with.

\begin{corollary}
	The SAT encoding $\Phi^{\B'}_S$ is satisfiable, if and only if there exists a program tree $\T$ with size $|\T|\leq\lfloor c/|\B'|\rfloor$ accepted by $\B'$.
\end{corollary}

\subsubsection{Size of construction}
The automaton can be constructed of size $O(2^{|B|+|\varphi|})$, i.e., for a fixed set of Boolean variables the automaton is linear in the size of the specification automaton or exponential in the size of the specification formula. The constructed constraint system $\Phi^{\B'}_S$ is of size $O(|T|\cdot|\delta|\cdot|\Sigma_P|)$ with $x$ many variables, where $x\in O(|T|\cdot(|T|+|\Sigma_P|+|Q|\cdot log(|Q|\cdot|T|)))$. Note that $|\Sigma_P|\in O(|B|^{N_\I+N_\O})$ grows polynomial in the number of variables for fixed input/output arities.

\section{Two-Wayless Encoding}
\label{twlencoding}
Next, we sketch the second encoding that avoids the detour
via universal two-way automata.
To this end, we alter the construction in that input- and output-labels collapse to a single \textit{InOut}-label with semantics as follows
\vspace{-0.2em}
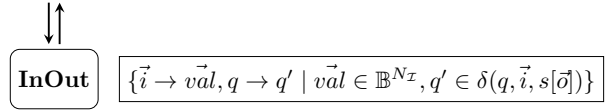
\begin{center}
  \scalebox{.9}{\begin{tikzpicture}[auto,semithick
					virtual/.style={thick,densely dashed},
					active/.style={thick,->,shorten >=2pt,shorten <=2pt,>=stealth},
					inactive/.style={color=lightgray,->,shorten >=2pt,shorten <=2pt,>=stealth},
					none/.style={draw=none}
					]
					\tikzstyle{every state}=[rectangle,rounded corners,fill=none,draw=black,text=black,initial text=]
					\tikzstyle{block} = [rectangle,draw=black,text=black,initial text=, minimum size=.6cm]
					
					\node[state] 	(A) at (0,0) {\textbf{InOut}};
					\node[block]    [align=left, anchor=west,xshift=.3cm](B) at (A.east) {$\{\vec{i}\rightarrow \vec{val}, q\rightarrow q' \mid \vec{val}\in\mb{B}^{N_\I}, q'\in\delta(q,\vec{i},s[\vec{o}])\}$};
					
					\draw[active] ([xshift=.08cm]A.north) -- (0.08,1.2);
					\draw[active] (-0.08,1.2) -- ([xshift=-.08cm]A.north);
					
					\draw[none] (A.217) -- (-1.05,-.95);
					\draw[none] (-.95,-1.05) -- (A.233);
					
					\draw[none] (A.307) -- (.95,-1.05);
					\draw[none] (1.05,-.95) -- (A.323);
					\end{tikzpicture}}
\end{center}
\vspace{-1.2em}
where we use output variables $ \vec{o} $ and input variables
$ \vec{i} $ that correspond to inputs and outputs of the system,
respectively. In a nutshell, our new encoding consists of four parts:

\begin{enumerate}

\item The first part guesses the program and ensures syntactical
  correctness.

\item The second part simulates the program for every possible input
  from every reachable \textit{InOut}-labeled state until it again
  reaches the next \textit{InOut}-labeled state. Note that every such
  simulation trace is deterministic once the input, read at the initial
  \textit{InOut}-labeled state, has been fixed.

\item The third part extracts a simplified transition structure from
  the resulting execution graph, that consists of direct input labeled
  transitions from one \textit{InOut}-labeled state to the next one
  and output labeled states.

\item In the last part, this structure is then verified by a run graph
  construction that must satisfy the specification, given as universal co-Büchi
  automaton. To this end, we couple inputs on the edges with
  the outputs of the successor state to preserve the Mealy semantics
  of the program.
  
\end{enumerate}

The first part utilizes a similar structure as used for the previous
encoding and thus is skipped for convenience here. To simulate the
program in the second part, we introduce the notion of a
\textit{valuation}~$ v \in \mc{V} $, where
\begin{equation*}
\mc{V} = P \times \mb{B}^{B \smallsetminus \mathcal{I}} \times \{ L,
R, U \} \times \mb{B}
\end{equation*}
captures the current program state, the current values of all
non-input variables, the current direction, and the result of the
evaluation of the last Boolean expression, respectively. The
simulation of the program is then expressed by a finite execution
graph, in which, after fixing a
inputs~$ \vec{i} \in 2^{\mathcal{I}} $, every valuation points to a
successor valuation. This successor valuation is unique, except for
\textit{InOut}-labeled states, whose successor depends on the next
input to be read. The deterministic evaluation follows from the rules
of Figure~\ref{semantics} and selects a unique successor for every
configuration, accordingly.

In part three, this expression graph then is compressed into a
simplified transition structure. To this end, we need for every input
and \textit{InOut}-labeled starting valuation, the target
\textit{InOut}-labeled valuation that is reached as a result of the
deterministic evaluation. In other words, we require to find a
shortcut from every such valuation to the next one. We use an
inductive chain of constraints to determine this shortcut
efficiently. Remember that we only know the unique successor of every
valuation which only allows to make one step forward at a time. Hence,
we can store for every valuation and input a second shortcut
successor, using an additional set of variables, constrainted as
follows: if the evaluated successor is \textit{InOut}-labeled, then
the shortcut successor must be the same as the evaluated
one. Otherwise, it is the same as the shortcut successor of the
successor valuation, leading to the desired inductive
definition. Furthermore, to ensure a proper induction base, we use an
additional ranking on the valuations that bounds the
number of steps between two \textit{InOut} labeled valuations. This
annotation is realized in a similar fashion as in the previously
presented encoding.

With these shortcuts at hand, we then can extract the simplified
transition structure, which is verified using a standard run graph
encoding as used for classical bounded synthesis. Furthermore, we use
an over-approximation to bound the size of the struture and use a
reachability annotation that allows the solver to reduce the
constraints to those parts as required by the selected solution. The
size can, however, also be bound using an explicit bound that is set
manually.

Using this separation into four independent steps allows to keep the
encoding compact in size, and results in the previously promised
performance improvements presented in the next section.

\section{Experimental Results}
\begin{table}[t]
  \caption{Comparison of the general and the two-wayless encoding.}
  \label{results}
	\centering
	\begin{tabular}{l|c|c|c|c|c}
		\hline
          \multicolumn{1}{c|}{\multirow{2}{*}{\, \textsc{specification} \,}}
          & \multirow{2}{*}{\, \textsc{states} \,}
          & \, \textsc{additional} \,
          & \multirow{2}{*}{\, $|\B'|$ \,}
          & \textsc{two-way}
          & \, \textsc{two-wayless} \, 
          \\
          && \textsc{variables} && \, \textsc{encoding} \, & \textsc{encoding} \\
	\hline \hline
		\textit{in} $\leftrightarrow$ \textit{out} & 6 & 0 & 16 & 00m16s & 00m02s \\
		\textit{in} $\leftrightarrow\LTLnext$ \textit{out} & 9 & 1 & 64 & 11m29s & 08m34s\\
		latch & 10 & 0 & 64 & $>$120m & 08m07s \\
          2-bit arbiter & 10 & 0 & 128 & 66m48s & 14m18s \\
          \hline
	\end{tabular}
\end{table}

\cref{results} compares the general encoding of \cref{general} and the
two-wayless encoding of \cref{twlencoding} on a selection of standard
benchmarks. The table contains the of number of states of the
program's syntax tree, the number of additional variables, i.e.,
variables that are not designated to handle inputs and outputs, the
size of the two-way universal Streett automaton, created for the
general encoding, and the solving times for both encodings.
\cref{implementations} shows the results in terms of the synthesized
program trees for the two-wayless encoding. The experiments indicate
a strong advantage of the second approach.

\begin{table}[t]
  \lstset{frame=none}
	\caption{Synthesized implementations for the two-wayless encoding.}
	\centering
        \label{implementations}
        \begin{tabular}{c|c|c|c}
          \hline
          \textit{in} $ \leftrightarrow $ \textit{out} & \textit{in} $ \leftrightarrow \LTLnext $ \textit{out} & latch & 2-bit arbiter \\
          \hline\hline
	\begin{minipage}{0.23\textwidth}
		\begin{lstlisting}[mathescape=true]
while(tt) {
  out = in;
  InOut
}
€\phantom{=}€
€\phantom{=}€
€\phantom{=}€
€\phantom{=}€
\end{lstlisting}
\end{minipage}
          &
          \ \begin{minipage}{0.23\textwidth}
		\begin{lstlisting}[mathescape=true]
while(tt) {
  out = var;
  var = in;
  InOut
}
€\phantom{=}€
€\phantom{=}€
€\phantom{=}€
		\end{lstlisting}
              \end{minipage}&
          \ \begin{minipage}{0.23\textwidth}
		\begin{lstlisting}[mathescape=true]
while (tt) {
  if (upd) {
    out = in
  } else {
    skip
  };
  InOut
}
		\end{lstlisting}
              \end{minipage} &
          \ \begin{minipage}{0.23\textwidth}
		\begin{lstlisting}[mathescape=true]
while (tt) {
  g0 = g1;
  g1 = not g1;
  InOut
}
€\phantom{=}€
€\phantom{=}€
€\phantom{=}€
		\end{lstlisting}
              \end{minipage} \\                             
                               \hline
	\end{tabular}        
\end{table}

\section{Conclusions}
We introduced a generalized approach to bounded synthesis that is applicable whenever all possible runs of a universal automaton on the possibly produced input/output words of an input-deterministic implementation can be expressed by a run graph. The acceptance of an implementation can then be expressed by the existence of valid annotation functions for an annotation comparison relation that expresses the acceptance of the automaton for Büchi, co-Büchi and Streett acceptance conditions. The existence of valid annotation functions for a run graph is encoded as a SAT query that is satisfiable if and only if there exists an implementation satisfying a given bound that is accepted by the automaton.

For LTL specifications, we constructed a two-way universal Streett automaton which accepts reactive programs that satisfy the specification. We then constructed a run graph that represents all possible runs and applied the generalized bounded synthesis approach. Next, we constructed a SAT query that guesses a reactive program of bounded size as well as valid annotation functions that witnesses the correctness of the synthesized program.

Finally, we merged the previous transformations into an extended encoding that simulates the program direclty via the constraint solver. We evaluated both encodings with the clear result that the encoding avoiding the explicit run graph construction for two-way automata wins in the evaluation.	

\bibliographystyle{splncs}
\bibliography{bsorp}

\newpage

\appendix

\section{Appendix}

\subsection{Equivalence of $\A$ and $\B$}\label{appendixEquivalenceProof}
To prove the equivalence, i.e., $\mc{L}(\A) = \mc{L}(\B)$, we need the formal definition of acceptance of trees by two-way tree automata and of the two-way automata $\A$ and $\B$. We fix the specification to be given as a non-deterministic Büchi word automaton $A_{spec} = (\Sigma, Q_{spec}, q_0, \delta_{spec}, \text{BÜCHI}(F_{spec}))$.

\subsubsection{Acceptance of two-way tree automata}\label{explicitA}

The acceptance is defined over run trees as follows. A run of a two-way alternating tree automaton $\A=(\Sigma,P ,p_0,\delta,Acc)$ on a program tree $\T = \tuple{T, \tau}$ is an infinite labeled binary tree $\T_R=\tuple{T_R, \tau_R}$, where $\tau_R: T_R \rightarrow T \times P \times \{L,R,D \}$ is a labeling function, which labels are tuples that contain the current state of the automaton and program tree as well as the direction it came from. For $\T_R$ it holds that
\begin{itemize}
	\item $\epsilon \in T_R$, $\tau_R(\epsilon)=(\epsilon, p_0, D)$ and
	\item $\forall y\in T_R$ with $\tau_R(y) = (t, p, d)$ and $\delta_t (p, \tau(t), d)=\theta:$\\
	Let $S=\big\{ (p_1,d_1),\dots,(p_n,d_n) \big\} \subseteq Q \times \{L,R,U\}$  be a set that satisfies $\theta$. Then $\forall 1\leq i\leq n: y \cdot i \in T_R$ and $\tau_R(y \cdot i) = \big( t', p_i, d' \big)$ with $(t', d') = \mu (t, d_i)$.
\end{itemize}

\subsubsection{Two-way alternating Büchi tree automaton $\A$}
The constructed two-way alternating Büchi automaton $\A$ over the alphabet $\Sigma_P$ with the set of states $P^\A$,
$$ P^\A_{expr} = S \times Bool $$
$$ P_{exec} = S \times Q_{spec} \times \mb{B}^{N_\I} \times \{inp, out\} \times Bool $$
$$ P^\A =  P^\A_{expr} \cup P_{exec} $$
and initial state $p^\A_0 = (s_0, q_0, i_0, inp, 0)$, is defined with the following transitions, where $s\in S$, $v\in Bool$, $q\in Q_{spec}$, $i\in \mb{B}^{N_\I}$, $m\in\{inp,out\}$ and $t\in\{0,1\}$.
\begin{itemize}
	\item \textbf{Transitions to evaluate Boolean expressions:}
	\begin{itemize}
		\item 
		$ \delta ^\A _{\emptyset} ((s,1), \textbf{tt}, D) = true \\
		\delta ^\A _{\emptyset} ((s,0), \textbf{tt}, D) = false$
		\item 
		$ \delta ^\A _{\emptyset} ((s,1), \textbf{ff}, D) = false \\
		\delta ^\A _{\emptyset} ((s,0), \textbf{ff}, D) = true$
		\item 
		$\delta ^\A _{\emptyset} ((s,v), b, D) = \begin{cases}true &,\text{if } s[b]=v \\
		false &,\text{otherwise}\end{cases}
		$
		\item 
		$ \delta ^\A _{LR} ((s,1), \vee, D)= ((s,1), L) \vee ((s,1), R) \\
		\delta ^\A _{LR} ((s,0), \vee, D)= ((s,0), L) \wedge ((s,0), R)$
		\item $ \delta ^\A _L ((s,v), \neg, D)= ((s,1-v),L)$
	\end{itemize}
	
	\item \textbf{Transitions to evaluate non I/O statements:}
	\begin{itemize}
		\item $ \delta ^\A _{\emptyset} ((s,q,i,m,t),\textbf{skip},D) = ((s,q,i,m,0), U)$
		\item $ \delta ^\A _L ((s,q,i,m,t),assign_b,D) =\\
		\big(((s[b/0],q,i,m,0),U) \wedge((s,0),L)\big)\;\vee\; \big( ((s[b/1],q,i,m,0),U) \wedge ((s,1),L)\big)$
		\item $ \delta ^\A _{LR} ((s,q,i,m,t),\textbf{if},D)=\\
		\big(((s,1),L) \wedge ((s,q,i,m,0), LR)\big) \; \vee \; \big(((s,0),L) \wedge ((s,q,i,m,0), RR)\big)$
		\item $ \delta ^\A _{LR} ((s,q,i,m,t),\textbf{while},D)\\
		=\delta ^\A _{LR} ((s,q,i,m,t),\textbf{while},R)\\
		= \big(((s,1),L) \wedge ((s,q,i,m,0), R)\big) \; \vee \; \big(((s,0),L) \wedge ((s,q,i,m,0), U)\big)$
	\end{itemize}
	\item \textbf{Transitions to evaluate input and output:}
	\begin{itemize}
		\item $\delta ^\A _{\emptyset} ((s,q,i,inp,t),\textbf{input }\vec{b},D) = \bigvee _{\vec{val}\in\mb{B}^{N_\mc{I}}}((s[\vec{b}/\vec{val}],q,\vec{val},out,0),U)$
		\item $\delta ^\A _{\emptyset} ((s,q,i,out,t),\textbf{output }\vec{b},D) = \bigvee _{q' \in \delta _{spec}(q,i,s[\vec{b}])} ((s,q',i,inp,1),U)$
	\end{itemize}
	\item \textbf{Transitions to move to next statement in program:}
	\begin{itemize}
		\item $ \delta ^\A _{LR} ((s,q,i,m,t),\textbf{;},D)= ((s,q,i,m,t),L)\\
		\delta ^\A _{LR} ((s,q,i,m,t),\textbf{;},L)= ((s,q,i,m,t),R)\\
		\delta ^\A _{LR} ((s,q,i,m,t),\textbf{;},R)= ((s,q,i,m,t),U) $
		\item $ \delta ^\A _{LR} ((s,q,i,m,t),\textbf{then},L)\\
		=\delta ^\A _{LR} ((s,q,i,m,t),\textbf{then},R)\\
		=((s,q,i,m,t),U)$
		\item $ \delta ^\A _{LR} ((s,q,i,m,t),\textbf{if},R) =((s,q,i,m,t),U) $
	\end{itemize}
\end{itemize}
All other transitions evaluate to $false$. The set of accepting states is defined as $$F^\A = \big\{ (s,q,i,m,1) \mid q \in F_{spec} \big\}$$

\subsubsection{Two-way non-deterministic Büchi tree automaton $\B$}\label{explicitB}

Our automaton $\B$ with the set of states
$$ P^\B_{expr} = S \times Q_{spec} \times \mb{B}^{N_\I} \times \{inp, out\} \times \{ \top , \bot \}$$
$$ P^\B =  P^\B_{expr}\cup P_{exec} $$
and initial state $p^\B_0 = (s_0, q_0, i_0, inp, 0)$, is defined with the following transitions, where $b\in B$, $s\in S$, $q\in Q_{spec}$, $i\in \mb{B}^{N_\I}$, $m\in\{inp,out\}$, $r\in\{\top,\bot\}$ and $t\in\{0,1\}$.

\begin{itemize}
	\item \textbf{Transitions to evaluate Boolean expressions:}
	\begin{itemize}
		\item
		$ \delta ^\B _{\emptyset} ((s,q,i,m,\bot), \textbf{tt}, D) = ((s,q,i,m,\top), U)$\\
		\item
		$ \delta ^\B _{\emptyset} ((s,q,i,m,\bot), \textbf{ff}, D) = ((s,q,i,m,\bot), U)$\\
		\item
		$\delta ^\B _{\emptyset} ((s,q,i,m,\bot), b, D) = \begin{cases}((s,q,i,m,\top), U) &,\text{if } s[b]=1 \\
		((s,q,i,m,\bot), U) &,\text{otherwise}\end{cases}$
		\item
		$ \delta ^\B _{LR} ((s,q,i,m,\bot), \vee, D)= ((s,q,i,m,\bot), L)
		\smallskip \\ 
		\delta ^\B _{LR} ((s,q,i,m,r), \vee, L)= \begin{cases}((s,q,i,m,r), U) &,\text{if } r = \top \\
		((s,q,i,m,r), R) &,\text{otherwise}\end{cases}\\
		\delta ^\B _{LR} ((s,q,i,m,r), \vee, R)= ((s,q,i,m,r),U)$
		\item $ \delta ^\B _L ((s,q,i,m,\bot), \neg, D)= ((s,q,i,m,\bot),L)\\
		\delta ^\B _L ((s,q,i,m,r), \neg, L)= ((s,q,i,m,\overline{r}),U)$ 
	\end{itemize}
	
	\item \textbf{Transitions to evaluate non I/O statements:}
	\begin{itemize}
		\item $ \delta ^\B _{\emptyset} ((s,q,i,m,t),\textbf{skip},D) = ((s,q,i,m,0), U)$
		\item $ \delta ^\B _L ((s,q,i,m,t),assign_b,D) = ((s,q,i,m,\bot),L)\\
		\delta ^\B _L ((s,q,i,m,r),assign_b,L) = ((s[b/r],q,i,m,0),U)$
		\item $ \delta ^\B _{LR} ((s,q,i,m,t),\textbf{if},D)= ((s,q,i,m,\bot),L)\\
		\delta ^\B _{LR} ((s,q,i,m,r),\textbf{if},L) = \begin{cases}((s,q,i,m,0), RL) &,\text{if } r = \top \\
		((s,q,i,m,0), RR) &,\text{otherwise}\end{cases}$
		\item $ \delta ^\B _{LR} ((s,q,i,m,t),\textbf{while},D)\\
		=\delta ^\B _{LR} ((s,q,i,m,t),\textbf{while},R)\\
		= ((s,q,i,m,\bot),L)\\
		\delta ^\B _{LR} ((s,q,i,m,r),\textbf{while},L) = \begin{cases}((s,q,i,m,0), R) &,\text{if } r = \top \\
		((s,q,i,m,0), U) &,\text{otherwise}\end{cases}$
	\end{itemize}
	\item \textbf{Transitions to evaluate input and output:}
	\begin{itemize}
		\item $\delta ^\B _{\emptyset} ((s,q,i,inp,t),\textbf{input }\vec{b},D) = \bigvee _{\vec{val}\in\mb{B}^{N_\mc{I}}}((s[\vec{b}/\vec{val}],q,\vec{val},out,0),U)$
		\item $\delta ^\B _{\emptyset} ((s,q,i,out,t),\textbf{output }\vec{b},D) =\bigvee _{q' \in \delta_{spec}(q,i,s[\vec{b}])} ((s,q',i,inp,1),U)$
	\end{itemize}
	\item \textbf{Transitions to move to next statement in program:}
	\begin{itemize}
		\item $ \delta ^\B _{LR} ((s,q,i,m,t),\textbf{;},D)= ((s,q,i,m,0),L)\\
		\delta ^\B _{LR} ((s,q,i,m,t),\textbf{;},L)= ((s,q,i,m,0),R)\\
		\delta ^\B _{LR} ((s,q,i,m,t),\textbf{;},R)= ((s,q,i,m,0),U) $
		\item $ \delta ^\B _{LR} ((s,q,i,m,t),\textbf{then},L)\\
		=\delta ^\B _{LR} ((s,q,i,m,t),\textbf{then},R)\\
		=((s,q,i,m,t),U)$
		\item $\delta ^\B _{LR} ((s,q,i,m,t),\textbf{if},R)= ((s,q,i,m,0),U)$\\
	\end{itemize}
\end{itemize}
All other transitions evaluate to $false$. The set of accepting states is defined as $$F^\B = \big\{ (s,q,i,m,1) \mid q \in F_{spec} \big\}$$

Note that $\B$ behaves similar to $\A$ during normal execution and that only boolean evaluation was altered. Therefore, the state spaces of the automaton only differs in the states corresponding to boolean evaluation. Additionally, the set of accepting states $F^\A$ and $F^\B$ are equivalent.

\subsubsection{Equivalence Proof}

In this section we prove that $\A$ and $\B$ accept the exact same program trees. We do so by showing first that $\B$ evaluates boolean expressions to the same result as $\A$, but instead of terminating with $true$ or $false$ the automaton walks the tree up again and eventually reaches the node where the boolean evaluation was started from containing the correct result. Combined with the fact that the automata have the same behavior on normal executions we show the equivalence of $\A$ and $\B$.

\begin{lemma}\label{lemma1}
	Given a program tree $\T$ that contains a boolean expression subtree with root $t\in T$. Let $\A$ be in state $(s,v)\in P^\A_{expr}$, $\B$ be in state $(s,q,i,m,\bot)\in P^\B_{expr}$ and both automata are reading program node $t$ coming down from the parent.
	
	$\A$ evaluates for state $(s,1)$ to $true$ and state $(s,0)$ to $false$, if and only if $\B$ moves upwards from node $t$ with state $(s,q,i,m,\top)\in P^\B_{expr}$. Vice versa, $\A$ evaluates for state $(s,1)$ to $false$ and state $(s,0)$ to $true$, if and only if $\B$ moves upwards from node $t$ with state $(s,q,i,m,\bot)\in P^\B_{expr}$.
\end{lemma}

\begin{proof}
	Let $\A$, $\B$ be defined as above, $\T$ be a valid program tree, $t\in T$ be the root of an boolean expression subtree. We proof the statement by structural induction over the boolean expression.\\
	Base-cases: For $\tau(t):$
	\begin{itemize}
		\item $\textbf{tt :}$\\
			$\delta ^\A _{\emptyset}\big( (s,1), \textbf{tt}, D \big) = true$ and  $\delta ^\A _{\emptyset}\big( (s,0), \textbf{tt}, D \big) = false$\\
			$\Leftrightarrow \delta ^\B _{\emptyset}\big( (s,q,i,m,\bot), \textbf{tt}, D \big) = ((s,q,i,m,\top), U)$
		\item $\textbf{ff :}$\\
			$\delta ^\A _{\emptyset}\big( (s,1), \textbf{ff}, D \big) = false$ and  $\delta ^\A _{\emptyset}\big( (s,0), \textbf{ff}, D \big) = true$\\
			$\Leftrightarrow \delta ^\B _{\emptyset}\big( (s,q,i,m,\bot), \textbf{ff}, D \big) = ((s,q,i,m,\bot), U)$
		\item $b\textbf{ :}$\\
		$\begin{array}{ll} \delta ^\A _{\emptyset} ((s,v), b, D) &=\begin{cases} true &,\text{if } s[b]=v \\
		false &,\text{otherwise}\end{cases}\\
		\Leftrightarrow \delta ^\B _{\emptyset} ((s,q,i,m,\bot), b, D) &= \begin{cases}((s,q,i,m,\top), U) &,\text{if } s[b]=1 \\
		((s,q,i,m,\bot), U) &,\text{otherwise}\end{cases} \end{array}$
	\end{itemize}
	Inductive-step: For $\tau(t)=$
	\begin{itemize}
		\item $\neg\textbf{ :}$\\
		By definition, $\A$ continues with $(s,1-v)$ and $\B$ continues with $(s,q,i,m,\bot)$ in the left subtree. By induction, $\B$ returns from the left subtree with state $(s,q,i,m,r)$ for $r=\top/\bot$ and by definition sends $(s,q,i,m,\overline{r})$ upwards, where $r$ matches the acceptance of $(s,v-1)$ and $\overline{r}$ with $(s,v)$.
		\item $\vee\textbf{ :}$\\
		By definition, $\B$ continues with $(s,q,i,m,\bot)$ in the left subtree.
		We make a case distinction based on $v$.\\
		Either $v = 1$: Then by definition, $\A$ evaluates to true, iff $(s,1)$ if either the left or right subtree evaluates to $true$.
		Assuming the left subtree evaluates to $true$, then by induction $\B$ returns from the left subtree with state $(s,q,i,m,\top)$ and by definition sends the same state upwards. For the other case, the right subtree evaluates to $true$ for $(s,1)$. By induction $\B$ first returns with $(s,q,i,m,\bot)$ from the left subtree and by definition sends this state into the right subtree. Again by induction, $\B$ returns with state $(s,q,i,m,\top)$ from the right subtree and by definition sends this state upwards.\\
		Or $v = 0$: Then by definition, $\A$ evaluates to $true$, iff $(s,0)$ in both the left and right subtree evaluate to $true$. By induction, $\B$ returns with copy $(s,q,i,m,\bot)$ from the right subtree and by definition sends this copy into the right subtree. Again by induction, $\B$ returns with $(s,q,i,m,\bot)$ and by definition sends it upwards.
	\end{itemize} \qed
\end{proof}

\begin{theorem}
	$\mathcal{L}(\A) = \mathcal{L}(\B)$.
\end{theorem}
\begin{proof}
	$"\Implies":$\\
	Let $\T\in\mathcal{L}(\A)$ be a program tree with accepting run tree $\R_\A = \langle R_\A, \tau_\A \rangle$. The initial state $p_0$ of $\A$ is a normal execution state ($p_0 \in P_{exec}$). The transitions of $\A$ are defined such that whenever $\A$ is in a normal execution state $p\in P_{exec}$ exactly one successor is again a normal execution state $p'\in P_{exec}$. There can be an additional successor state that is used to evaluate a boolean expression but those subtrees of the accepting run graph are finite and terminate to $true$ on every path. Therefore there only exists a single infinite path $$r_0r_1\hdots = (t_0, p_0, d_0)(t_1, p_1, d_1)\hdots \in \big(T\times P_{exec}\times\{L,R,D\}\big)^\omega$$ in the accepting run tree $\R_\A$ that satisfies the Büchi acceptance condition.
	
	We show that $\B$ visits the same states $p_0p_1\hdots\in P_{exec}^\omega$ in the same order while traversing $\T$, that is whenever $\A$ is in state $p_i\in P_{exec}$ reading a node $t_i\in\T$ moving to node $t_{i+1}$ with state $p_{i+1}$, $\B$ in the same state reading the same node eventually moves to node $t_{i+1}$ with state $p_{i+1}$. We show this with a case analysis over the possible labels $\tau_\A(t_i)$. In all three interesting cases $r_i$ has two children, namely $r_{i+1}$ and $(t',p',L)$ with $p'=(s,v)\in P^\A_{expr}$.
	\begin{itemize}
		\item $assign_b$: Since the run tree is accepting, the boolean evaluation evaluates to $true$. With \cref{lemma1} we know that $\B$ traverses the boolean expression subtree and returns with state $(s,q,i,m,r)$ to $t_i$, where $r=\top/\bot$ for $v=1/0$, respectively. By definition $\B$ now sends $(s[b/v],q,i,m,0)$ upwards yielding $r_{i+1}$.
		\item $\textbf{if}$: Since the run tree is accepting, the boolean evaluation evaluates to $true$. With \cref{lemma1} we know that $\B$ traverses the boolean expression subtree and returns with state $(s,q,i,m,r)$ to $t_i$, where $r=\top/\bot$ for $v=1/0$, respectively. Depending on $v$ having the value $0$ or $1$ the only valid move for $\B$ is either $((s,q,i,m,0),RR)$ or $((s,q,i,m,0),RL)$, respectively. Either way it matches the move of $\A$ yielding $r_{i+1}$.
		\item $\textbf{while}$: Since the run tree is accepting, the boolean evaluation evaluates to $true$. With \cref{lemma1} we know that $\B$ traverses the boolean expression subtree and returns with state $(s,q,i,m,r)$ to $t_i$, where $r=\top/\bot$ for $v=1/0$, respectively. Depending on $v$ having the value $0$ or $1$ the only valid transition for $\B$ is either $((s,q,i,m,0),U)$ or $((s,q,i,m,0),R)$, respectively. Either way it matches the move of $\A$ yielding $r_{i+1}$.
	\end{itemize}
	All other valid transitions are equally defined such that $\B$ can make the same move.
	
	Since $\B$ visits for every accepted run tree the same normal execution states equally often and the set of accepting states for $\A$ and $\B$ are equivalent and a subset of normal execution states, $\B$ accepts the same program trees as $\A$.
	
	$"\Leftarrow":$\\
	Let $\T\in\mathcal{L}(\B)$ be a program tree with accepting run tree $\R_\B = \langle R_\B, \tau_\B \rangle$. The initial state $p_0$ of $\B$ is a normal execution state ($p_0 \in P_{exec}$). Since $\B$ has no universal choices the run tree is 1-ary, that is a single infinite path $$r_0r_1\hdots=(t_0,p_0,d_0)(t_1,p_1,d_1)\hdots\in\big(T\times P^\B\times \{L,R,D\} \big)^\omega$$
	
	We show that $\A$ can visit the same normal execution states as $\B$, more precisely: Let $\B$ be in state $p_i\in P_{exec}$ reading a node $t_i\in\T$. If $\B$ moves to node $t_{i+1}$ with state $p_{i+1}\in P_{exec}$, then $\A$ can also perform a valid move from $t_i$ with state $p_i$ to $t_{i+1}$ with state $p_{i+1}$.\\ If $\B$ moves to a state $p_{i+1}\in P^\B_{expr}$ and traverses the boolean expression subtree and therefore visits the states $p_{i+2}p_{i+3}\hdots p_{i+j-1}\in(P^\B_{expr})^*$ until it returns to node $t_i(=t_{i+j})$ with state $p_{i+j}\in P^\B_{expr}$ and subsequently moves to node $t_{i+j+1}$ with state $p_{i+j+1}\in P_{exec}$, then $\A$ can also perform a valid move from $t_i$ with state $p_i$ to node $t_{i+j+1}$ with state $p_{i+j+1}$.
	
	In the first case, where $\B$ moves from normal execution state to another normal execution state $\A$ can always do the same move since those transitions are equally defined. We show the second case with a case analysis over the possible labels $\tau_\B(t_i)$. Let $p_i = (s,q,i,m,t)\in P_{exec}$, $t_i \in \T$, $p_{i+1} = (s,q,i,m,\bot)\in P^\B_{expr}$ and $p_{i+j}=(s,q,i,m,r))$.
	\begin{itemize}
		\item $assign_b$: In this case $\B$ moves from $p_{i+j}$ to $(s[b/v],q,i,m,0)$. Depending on $v$, $\A$ moves from $t_i$ with state $p_i$ into $t_{i+j+1}$ with state $(s[b/v],q,i,m,0)$ and universally into $t_{i+1}$ with state $(s,v)$, thus satisfying the transition relation. The second copy send into the boolean expression terminates to $true$ based on \cref{lemma1}.
		\item $\textbf{if}$ and $\textbf{while}$: $\A$ moves from $t_i$ with state $p_i$ into $t_{i+j+1}$ with state $p_{i+j+1}$ and depending on $v$ universally into $t_{i+1}$ with state $(s,v)$. Again with \cref{lemma1} the second copy send into the boolean expression terminates to $true$.
	\end{itemize}
	Since $\A$ visits the same normal execution states as $\B$ in the same order and the set of accepting states are equivalent and is a subset of normal execution states, $\A$ fulfills the Büchi acceptance condition. \qed
\end{proof}

\subsection{Proof of \cref{generalStreett}}\label{generalStrettProof}
\textit{For a Streett acceptance condition $Acc=\text{Streett}(F)$ with set of tuples of states $F\subseteq 2^{Q\times Q}$ and a run graph $\G=(V,v_0,E,f)$:\\
If $\G$ satisfies $Acc$, then there exists a valid $|V|$-bounded annotation function $\lambda$ for each basic comparison relation in $\triangleright^F_S$.}
\begin{proof}
	Let $F=\{(A_i,G_i)\}_{i\in[k]}$ be a set of tuples of states, $Acc=\text{STREETT}(F)$ an acceptance condition and $\G=(V,v_0,E,f)$ a run graph satisfying $Acc$.
	
	We construct the $k$ different valid $c$-bounded annotation functions $\lambda_i$. We first define $\lambda_i(v) = 0$ for all unreachable states $v\in V$ and then remove those states from $\G$. We then define each $\lambda_i(v)$ for all $i\in [k]$ with:
	\begin{itemize}
		\item $\forall v\in V$ with $f(v)\in G_i: \lambda_i(v)=0$.
		\item Remove outgoing edges for states $v\in V$ with $f(v)\in G_i$. That results in $\G'=(V,E',f)$, where $E'=E\setminus\{ (v,v')\in E\mid f(v)\in G_i\}$. $\G'$ contains no \textit{strongly connected components} (SCCs) that have a vertex $v$ with $f(v)\in A_i\cup G_i$. For $G_i$ it is obvious since all outgoing edges are removed. For $A_i$: Assume there is an SCC containing a state $v$ where $f(v)\in A_i$, then there exists an infinite path $\pi$ inside this SCC that infinitely often visits $A_i$ without infinite visits to the corresponding set $G_i$ and therefore $f(\pi)\not\in Acc$ which contradicts the assumption.
		\item Now let $\mc{S}$ be the set off all SCCs of $\G'$. We define $\G''=(V'',E'',f)$. Where $V''=\mc{S}\cup\{ \{v\}\mid v\not\in \bigcup_{S\in\mc{S}}S \}$ and $E''=\{ (S_1,S_2)\mid \exists v_1\in S_1, v_2\in S_2: S_1\not=S_2 \wedge (v_1,v_2)\in E' \}$.\\
		$\G''$ is a directed acyclic graph and therefor all paths in $\G''$ are of finite length and there only exist finitely many paths.
		\item We define $nb_i(\pi):(V'')^*\rightarrow \mb{N}$ as the max number of bad states visited on some path $\pi=S_1,\hdots,S_m$. Let $Occ(\pi)$ denote the set of occurring SCCs in $\pi$. We formally define $nb_i(\pi)=\big|\bigcup_{S\in Occ(\pi)}S \cap \{ v\in V\mid f(v)\in A_i \}\big|$.
		\item $\forall v\in S \in V''$ with $f(v)\not\in G_i: \lambda_i(v)=\max(\{nb_i(\pi)\mid \pi \text{\textit{ is a path from }}S\})$.
	\end{itemize}\qed
\end{proof}

\end{document}